\documentclass[smallextended]{svjour3}
\smartqed 
\usepackage{graphicx}

\usepackage[utf8]{inputenc} % allow utf-8 input
\usepackage[T1]{fontenc}    % use 8-bit T1 fonts
\usepackage{hyperref}       % hyperlinks
\usepackage{url}            % simple URL typesetting
\usepackage{booktabs}       % professional-quality tables
\usepackage{amsfonts}       % blackboard math symbols
\usepackage{nicefrac}       % compact symbols for 1/2, etc.
\usepackage{microtype}      % microtypography
\usepackage{lipsum}
\usepackage{graphicx}
\usepackage{amsmath,amssymb,amsfonts}
\usepackage{algorithmic}
\usepackage{graphicx}
\usepackage{textcomp}
\usepackage{xcolor}
\usepackage{dsfont}
\usepackage{enumitem}
\usepackage{subcaption}
\usepackage[ruled,vlined]{algorithm2e}
\usepackage{tikz}
\usepackage{onimage}
\usepackage{enumitem}

\usepackage{natbib}
\setcitestyle{authoryear,open={(},close={)}}

\newcommand{\RN}[1]{%
	\textup{\uppercase\expandafter{\romannumeral#1}}%
}

\newcommand{\bs}[1]{\boldsymbol{#1}}
\newcommand{\R}{{\mathds{R}}}

\newcommand{\ES}{{\mathbf S}}

\newcommand{\Na}{{N_{\rm a}}}

\newcommand{\No}{{N_{\rm o}}}

\newcommand{\Zz}{\bs{Z}}

\newcommand{\Am}[4]{A_{#1#2}^{#3#4}}

\newcommand{\af}[1]{{\color{black}{#1}}}

\begin{document}

\title{Analysis and control of agreement and disagreement opinion cascades\thanks{This research has been supported in part by NSF grant CMMI-1635056, ONR grant N00014-19-1-2556, ARO grant W911NF-18-1-0325, DGAPA-UNAM PAPIIT grant IN102420, Conacyt grant A1-S-10610, and by NSF Graduate Research Fellowship DGE-2039656. Any opinion, findings, and conclusions or recommendations expressed in this material are those of the authors and do not necessarily reflect the views of the NSF.}
}


\author{Alessio Franci         \and
        Anastasia Bizyaeva     \and
        Shinkyu Park		   \and
        Naomi Ehrich Leonard
}


\institute{Alessio Franci (Corresponding Author) \at
              Math Department\\
              National Autonomous University of Mexico\\
              04510 Mexico City, Mexico\\
              \email{afranci@ciencias.unam.mx}     
           \and
           Anastasia Bizyaeva, Shinkyu Park, Naomi Ehrich Leonard \at
           Department of Mechanical and Aerospace Engineering\\
           Princeton University\\
           Princeton, NJ 08544\\
           \email{bizyaeva@princeton.edu, shinkyu@princeton.edu, naomi@princeton.edu}
        %   \and
        %   Shinkyu Park \at
        %   Department of Mechanical and Aerospace Engineering\\
        %   Princeton University\\
        %   Princeton, NJ 08544\\
        %   \email{shinkyu@princeton.edu}
        %   \and
        %   Naomi Ehrich Leonard \at
        %   Department of Mechanical and Aerospace Engineering\\
        %   Princeton University\\
        %   Princeton, NJ 08544\\
        %   \email{naomi@princeton.edu}
}

\date{Received: date / Accepted: date}

\maketitle

\begin{abstract}
We introduce and analyze a continuous time and state-space model of opinion cascades on networks of large numbers of agents that form opinions about two or more options.
By leveraging our recent results on the emergence of agreement and disagreement states, we introduce novel tools to analyze and control agreement and disagreement opinion cascades. New notions of agreement and disagreement centrality, which depend only on network structure, are shown to be key to characterizing the nonlinear behavior of agreement and disagreement opinion formation and cascades. Our results are relevant for the analysis and control of opinion cascades in real-world networks, including biological, social and artificial networks, and for the design of opinion-forming behaviors in robotic swarms. We illustrate an application of our model %in a detailed 
to a multi-robot task-allocation problem
%example and thoroughly 
and discuss extensions and future directions opened by our modeling framework.

\end{abstract}

\section{Introduction}
\label{sec:Intro}

Network cascades are classically modeled as the spread of a binary transition from an inactive (e.g., uninformed, not infected) to an active (e.g., informed, infected) state. Models have been used to examine a variety of phenomena, including the spread of fads~\citep{gladwell2006,bikhchandani1992}, competing technologies~\citep{arthur1989competing}, diffusion of innovations~\citep{valente1995network,rosa2013non}, and the spread of influence~\citep{kempe2003maximizing}. The linear threshold model (LTM), introduced in~\cite{granovetter1978threshold,schelling2006micromotives}, 
defines a basic modeling hypothesis for studying network cascades, namely, discrete time and binary states of agents that update as a threshold function of the states of neighbors~\citep{watts2002simple,lim2015simple,garulli2015analysis,rossi2017threshold}. Recently, the LTM was extended to  multi-layer interaction networks~\citep{yaugan2012analysis,salehi2015spreading,zhong2020influence}. % but also in these works the agents' states are assumed to be binary and updating in discrete time.

In an opinion formation context the situation is more general. When agents form opinions about a set of options, the transition from an unopinionated to an opinionated state is not a binary process because the opinionated state can correspond to a variety of distinct opinion configurations. Even in the simplest case of two options, the opinionated state can evolve in two qualitatively distinct directions, that is, preferring option A or preferring option B. Thus, for opinion cascades, not only does the active or inactive nature of an agent's neighbors matter but also the opinion state of neighbors, i.e., their preferences among options. That is, the opinion state to which an agent will converge during an opinion cascade  depends not only on the agent's susceptibility to forming an opinion (i.e., its activation threshold) but also the agent's tendency to align with, or reject, its neighbors' opinions. Further, %opinions have an intrinsically nuanced nature. 
the transition from being unopinionated to holding a strong opinion in favor or disfavor of one or more options can be more or less continuous, both in time and in strength, and the graded transition from one opinion state to another cannot be captured in binary models. A continuous time and state-space approach to threshold cascade behavior was explored in~\cite{zhong2019continuous}. A multi-option, agent-based approach to opinion cascades was studied in~\cite{pilditch2017opinion}. The approach followed in the present paper considers both the presence of multiple options and multiple possible opinion states and a process that is continuous both in time and in state.

In recent work~\citep{bizyaeva2020general}, we developed a general model of continuous time and state-space opinion formation for an arbitrary number of agents exchanging information over a communication network and forming opinions about an arbitrary number of options. Our model is both analytically tractable and general because it is motivated by a tractable and general model-independent theory of opinion formation~\citep{AF-MG-AB-NEL:20}. In particular, very general opinion-formation behaviors like consensus, dissensus, and various intermediate forms of agreement and disagreement, can all be realized and controlled in our model with only a handful of key parameters with clear real-world interpretations (e.g., parameters that represent the agents' level of cooperativity or competitivity). 

In our model, opinion formation happens through a bifurcation.  \af{While much can be investigated with linear models (see, e.g., \cite{Mateo2019}), the nonlinearity of our model allows for richer possibilities.}  Notably, the resulting opinion-formation process can be continuous or switch-like, depending on the nature of the bifurcation and the way the model crosses it~\citep{AF-MG-AB-NEL:20}. Further, when options are close to equally valuable, there is multistability of either agreement or disagreement opinion configurations. In this case, the opinion-forming bifurcation breaks deadlock by pushing the state away from the neutral (deadlock) equilibrium and towards one of the stable opinion states. The key  (bifurcation) parameter of our model is the ``attention'' $u_i$ that agent $i$ pays to its neighbors in the communication network. This parameter can also be interpreted as the urgency to form an opinion that the agent perceives. When agents have low attention, they form opinions linearly according to the information-rich inputs they receive about the available options. Agents that receive no inputs (i.e., those that do not gather information about the options) remain neutral or unopinionated. Conversely, when agents have high attention, they form a robust opinion through nonlinear network interactions, independently of, but also influenced by, their inputs.

In this paper, we leverage the state-dependent attention dynamics introduced in~\cite{bizyaeva2020general} to show that a positive feedback coupling between attention and opinion strength at the single-agent level provides the basic mechanism underlying the emergence of tunable opinion cascades at the network level. The %intra-agent 
positive feedback makes single-agent opinion formation switch-like and provides it with a tunable threshold, like in binary models of opinion cascade, while also preserving the continuous-time, graded evolution of the opinion and attention states.  We develop a detailed geometrical analysis of the single-agent attention-opinion dynamics for the case of three options, which extends the two-option geometrical analysis in~\cite{bizyaeva2020general}. The two cases have similarities (i.e., hysteresis and switches between unopinionated and opinionated states) but also  differences (i.e., the much richer menu of possible opinionated configurations for three options). 

Because of the intrinsic multistability of opinion configurations, our model naturally captures ``afterthoughts'' during the opinion cascade process: an early opinionated agent can change its opinion under the influence of sufficiently strong signals from other agents. Binary threshold models do not exhibit the same degree of behavioral flexibility. \af{Further, due to its continuous time and state-space nature, our model features both  simple and complex contagion (in the sense of~\cite{centola2007complex}). That is, along the same opinion cascade process, contagion can be simple for some agents and complex for others}.

To analyze and control the emergent behavior of the cascade model on undirected network topologies, we introduce the notions of {\it agreement centrality}, {\it disagreement centrality} and {\it signed disagreement centrality}. Agreement centrality coincides with the classical notion of eigenvector centrality~\citep{bonacich1972,bonacich2007some}. Disagreement and signed disagreement centralities are new. Our definitions are motivated by the tight connections existing between the patterns of opinion formation emerging at an opinion forming bifurcation and the spectral properties of the underlying adjacency matrix~\citep{bizyaeva2020patterns}. We show that our definitions provide powerful tools to analyze and control opinion cascades and their sensitivity to cascade seed inputs.

\af{Our opinion-formation and opinion-cascade modeling, analysis, and control tools are naturally suited for engineering, biological, and sociopolitical applications. For example, biological and artificial swarms in motion can make both agreement decisions, such as to move in the same direction like schooling fish~\citep{Couzin2005a,Couzin2011,Leonard2012}, and disagreement decisions, such as to go in different directions during search and foraging~\citep{traniello1989foraging,zedadra2017multi}. Sociopolitical networks, both at the citizen and political elite level, are made of agents that can more or less rapidly change between agreement and disagreement behaviors, depending on the topic, time and context~\citep{McCarty2019,macy2019opinion}.

Two main types of collective decision making, across applications, can broadly be classified as ``best of $n$" and ``task allocation"~\citep{valentini2017best,gerkey2004formal,Brambilla2013,schranz2020swarm}. In the former, agents must reach a consensus about the best-of-$n$ options, while in the latter, agents must allocate across the $n$ options. Thus, solving a best-of-$n$ problem requires a specific type of agreement opinion formation. Indeed, an earlier special case of our model was originally inspired by best-of-$n$ collective decision making in house hunting honeybees~\citep{Gray2018}.
%Including an attention-based cascade dynamics naturally allows us to capture the speed-accuracy trade-off of best-of-$n$ decision making.

Alternatively, solving a task-allocation problem requires a specific type of disagreement opinion formation. In this paper, we illustrate this fact with an original example describing distributed allocation over three tasks by a swarm of robots.  With this example, we highlight the importance of task-allocation cascades.  %this fact, including the importance of task-allocation cascades, through an original example describing a swarm of robots that must allocate themselves across three tasks in a distributed fashion. 
The result is a dynamical task-allocation algorithm controlled by just a few parameters. Given its continuous time and state-space nature, the algorithm is naturally suited to implementation in low-level (e.g., analog) hardware, making it relevant for any  multi-robot application with restricted computational power.

%Similarly, biological and artificial swarms in motion can make both agreement (i.e., following the same direction, like in schooling fish~\cite{Couzin2005a,Couzin2011,Leonard2012}) or disagreement (i.e., following different directions, like during search and foraging~\cite{traniello1989foraging,zedadra2017multi}) decisions. Finally, sociopolitical networks, both at the citizen and political elite level, are made of agents that can more or less rapidly change between agreement and disagreement behaviors, depending on the topic, time and context~\cite{McCarty2019,macy2019opinion}}

%Throughout the paper, we use randomly generated Watts-Strogatz~\cite{watts1998collective} and Barabasi-Albert~\cite{albert2002statistical} networks. These two types of networks are numerically shown to exhibit sharply different opinion forming and opinion cascade behaviors. The centrality indices rigorously capture the differences between the two networks. Our theoretical approach is general and can be applied to any network topology.
%¡¡¡(move this to the paper structure)!!!

The paper is structured as follows. In Section~\ref{SEC: general model}, we introduce the general opinion dynamics model of~\cite{bizyaeva2020general} and we rigorously analyze its opinion-forming behavior. In Section~\ref{sec:task allocation}, we %illustrate in a detailed example how 
frame a multi-robot task allocation problem %can be framed in 
using our opinion formation dynamics. In Section~\ref{sec: opinion cascades}, we show how an attention positive feedback mechanism at the agent and network scales engenders opinion cascade with tunable sensitivity in our opinion formation dynamics. To analyze and control opinion cascades, in Section~\ref{sec: cascade control} we introduce new notions of centrality and illustrate their power mechanistically in low-dimensional agreement and disagreement opinion cascade examples. In Section~\ref{sec:task allocation cascades}, we show how introducing a cascade mechanism makes the multi-robot task-allocation model introduced in Section~\ref{sec:task allocation} reactive to changes in task urgency and more robust to initial conditions and parameter uncertainties. Application to large scale complex networks is presented in Section~\ref{sec: complex networks}. Section~\ref{sec: discussion} provides a discussion of results, our notion of agreement and disagreement centralities, and the connection between our notion of opinion cascade and complex contagion. Section~\ref{sec: discussion} also describes  extensions and future research directions opened up by our modeling framework, including time-varying network topologies and communication losses, and  applications beyond multi-robot task allocation.}

%We further discuss the relevance of our modeling framework for real-world application, both from an opinion-formation and opinion-cascade viewpoint, in the discussion (Section...).

%...The dynamical behavior of the proposed algorithms follow our opinion-formation and cascade-control theory, as illustrated via extensive numerical simulations.

\section{A general model of agreement and disagreement opinion formation}
\label{SEC: general model}

We consider a network of $\Na$ agents forming opinions about $\No$ options. Let $z_{ij}\in\R$ denote the opinion of agent $i$ about option $j$. The greater $z_{ij}$ is, the stronger is the preference of agent $i$ for option $j$. We assume that
\begin{equation}\label{EQ:relative opinions}
z_{i1}+\cdots+z_{i\No}=0.  
\end{equation}
This assumption means that we are only interested in modeling an agent's {\it relative} preference for the various options. Indeed, to determine if agent $i$ prefers option $j$ over option $l$, $l\neq j$, only the value of $z_{ij}$ relative to $z_{il}$ matters. %, and not the values $z_{ij}$ and $z_{il}$. 
Condition~\eqref{EQ:relative opinions} can also be interpreted as a linearized simplex condition, as detailed in~\cite{AF-MG-AB-NEL:20}.

Let $\Zz_{i} =(z_{i1}, \dots ,z_{i N_{o}})$ be the {\em opinion state of agent $i$}  and $\Zz = (\Zz_{1}, \dots, \Zz_{N_{a}})$ the {\em opinion state of the agent network}. We recall some useful definitions from~\cite{AF-MG-AB-NEL:20,bizyaeva2020general}. Agent $i$ is {\it neutral} if $\Zz_i=\mathbf{0}$. The origin $\Zz=\mathbf{0}$ is called the network's {\it neutral state}. Agent $i$ is \textit{unopinionated} if its opinion state is small, i.e., $\| \Zz_{i}\| \leq \vartheta$, for a fixed threshold $\vartheta \geq 0$. Agent $i$ is {\it opinionated} if $\| \Zz_{i}\| \geq \vartheta$. If $z_{ij} \geq z_{il} - \vartheta$ for all $l \neq j$, then agent $i$ is said to {\it favor} option $j$. If $z_{ij} < z_{il} - \vartheta$ for all $l \neq j$, then agent $i$ is said to {\it disfavor} option $j$. An agent is {\it conflicted} about a set of options if it has near equal and favorable opinions of them all. 
Agents can agree and disagree. Two agents {\it agree} when they are opinionated and have the same qualitative opinion state (e.g., they both favor the same option). They {\it disagree} when they have qualitatively different opinion states. We refer the reader to~\cite{AF-MG-AB-NEL:20} for a geometric illustration of these definitions. The network is in an {\it agreement} state when all agents are opinionated and agree and in a {\it disagreement} state when at least two agents disagree.

Agents' opinions are assumed to evolve in continuous time. The networked opinion dynamics \citep{bizyaeva2020general} are the following:
\begin{subequations}\label{EQ:generic decision dynamics}
	\begin{align}
	& \quad \; \; \;\dot z_{ij}=F_{ij}(\Zz)-\frac{1}{\No}\sum_{l=1}^{\No}F_{il}(\Zz) \label{eq:modela}\\
	&F_{ij}(\Zz)=-d_{ij}z_{ij} + 
	u_{i}\left(  S_{1}\left( \sum_{\substack{k=1}}^\Na\Am{i}{k}{j}{j} z_{kj}\right)+ \sum_{\substack{l\neq j\\l=1}}^\No S_{2}\left( \sum_{k = 1}^\Na \Am{i}{k}{j}{l} z_{kl}\right)\right)+b_{ij},\label{eq:modelb}
	\end{align}
\end{subequations}
where $\dot z_{ij}$ is the rate of change, with respect to time $t$, of the opinion of agent $i$ about option $j$. Observe that, by~(\ref{eq:modela}), $\dot z_{i1}+\cdots+\dot z_{i\No}=0$, which ensures that $z_{i1}(t)+\cdots+z_{i\No}(t)=0$ for all $t>0$, provided $z_{i1}(0)+\cdots+z_{i\No}(0)=0$. The parameter $d_{ij}>0$ models the resistance of agent $i$ to form an opinion about option $j$. Exogenous inputs (e.g., prior biases or incoming information) influencing the opinion formation of agent $i$ about option $j$ are captured in the affine parameter $b_{ij}$. Network interactions are modeled in the {\it adjacency tensor} $A$. There are four types of network interactions \af{represented by elements of the adjacency tensor as follows}:
\begin{enumerate}
	\item intra-agent, same-option (with weights $\Am{i}{i}{j}{j}$)\af{, modeling the influence of agent $i$'s opinion about option $j$ on itself};
	\item intra-agent, inter-option (with weights $\Am{i}{i}{j}{l}$)\af{, modeling the influence of agent $i$'s opinion about option $l$ on the same agent's opinion about option $j$};
	\item inter-agent, same-option (with weights $\Am{i}{k}{j}{j}$)\af{, modeling the influence of agent $k$'s opinion about option $j$ on agent $i$'s opinion about the same option};
	\item inter-agent, inter-option (with weights $\Am{i}{k}{j}{l}$)\af{, modeling the influence of agent $k$'s opinion about option $l$ on agent $i$'s opinion about option $j$}.
\end{enumerate}
Network interactions can be {\it excitatory}, if the associated weight is positive, or {\it inhibitory}, if the associated weight is negative. $S_{q}: \mathds{R} \to [-k_{q1},k_{q2}]$ with $k_{q1},k_{q2} \in \mathds{R}^{>0}$ 
for $q \in \{1,2\}$ is a generic sigmoidal saturating function satisfying constraints  $S_{q}(0) = 0$, $S_{q} '(0) = 1$, $S_{q} ''(0) \neq 0 $, $S_{q}'''(0) \neq 0 $. The two sigmoids $S_1$ and $S_2$ saturate same-option interactions and inter-option interactions, respectively. They are the only nonlinearity of the model. 

The {\it attention parameter} $u_i$ tunes how strongly nonlinear network interactions affect agent $i$ (Figure~\ref{FIG: model}):
\begin{itemize}
	\item for small $u_i$, agent $i$'s opinion formation  is dominated by a linear behavior determined by $d_{ij}$ and $b_{ij}$;
	\item for large $u_i$, agent $i$'s opinion formation is dominated by a nonlinear behavior determined by  network interactions.
\end{itemize}
\begin{figure}
	\centering
	\includegraphics[width=\textwidth]{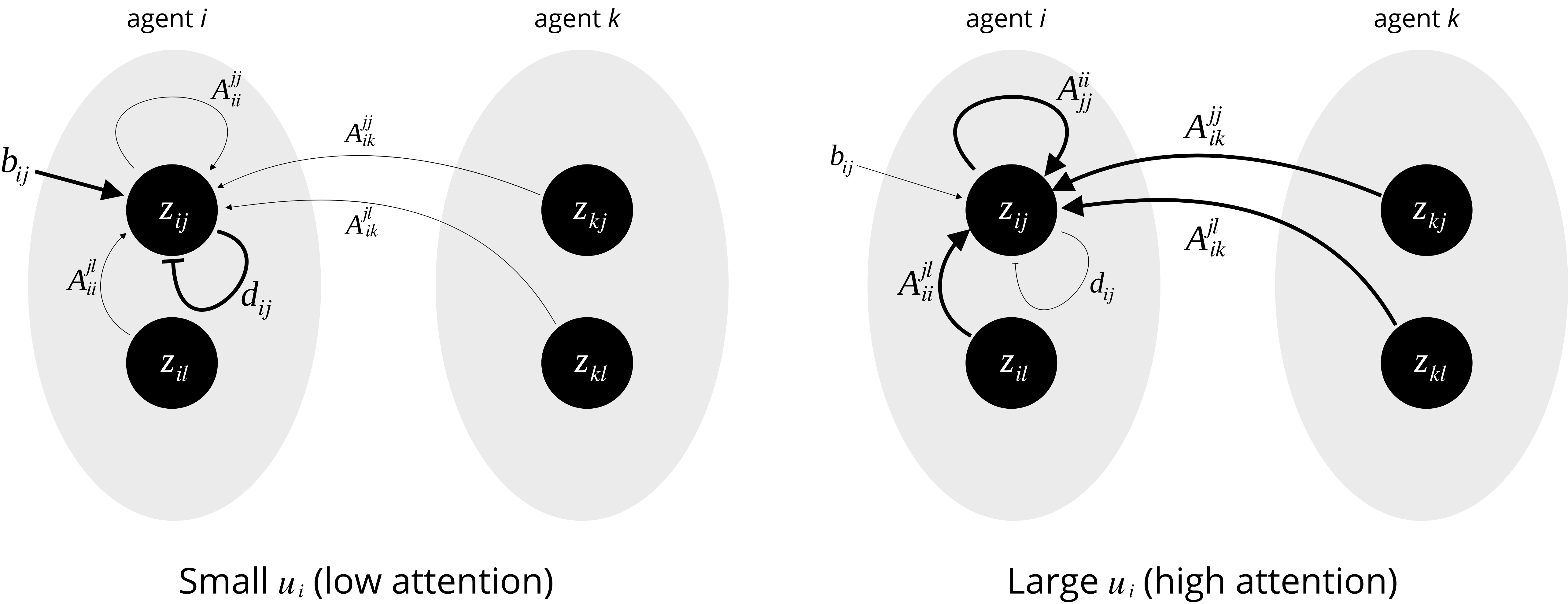}
	\caption{Schematic model description.}\label{FIG: model}
\end{figure}
In the linear regime, an agent develops an opinion proportional to the ratio of its inputs and its resistance, much like a classical leaky evidence integrator~\citep{bogacz2007optimal}. In the nonlinear regime, much richer network opinion formation behaviors can be expected. \af{Model~\eqref{EQ:generic decision dynamics} was inspired by and specializes to the honeybee-inspired model in~\cite{Gray2018}. See also~\cite{bizyaeva2020general} for a discussion of the similarities and difference between model~\eqref{EQ:generic decision dynamics} and previous opinion formation models in the literature.}

\subsection{Agreement and disagreement opinion formations are determined by network spectral properties}

Model~\eqref{EQ:generic decision dynamics} contains hundreds of intrinsic and interaction parameters when $\Na$ and/or $\No$ are sufficiently large. As such, it is intractable in general. However, the model-independent theory in~\cite{AF-MG-AB-NEL:20} shows that even in the simplest case in which all  four types of network interaction are all-to-all and homogeneous, inputs are zero, attention is uniform across the agents, and the resistance is the same for all the agents, a plethora of rich prototypical opinion formations are still possible in model~\eqref{EQ:generic decision dynamics}. 

To study the role of complex inter-agent network interactions,  while preserving analytical tractability, we consider the {\it homogeneous regime} of model~\eqref{EQ:generic decision dynamics}, defined by imposing the following restrictions on the model parameters:
\begin{align}\label{eq: homogeneous params}
d_{ij} &= d > 0,\ u_{i} = u,\ \Am{i}{i}{j}{j} = \alpha \in \mathds{R}\nonumber\\
\Am{i}{i}{j}{l} &= \beta \in \mathds{R},\ \Am{i}{k}{j}{j} = \gamma \tilde{a}_{ik},\ \Am{i}{k}{j}{l} = \delta \tilde{a}_{ik}
\end{align}
where $\alpha,\beta,\gamma,\delta \in \mathds{R}$ and $\tilde{a}_{ik} \in \{ 0, 1 \}$, for all $i,k = 1,\ldots,\Na$, $i \neq k$, and for all $j,l =1,\ldots,\No$, $j \neq l$. The resulting dynamics are
%\red{forgot to set the relative inputs to zero?  The results (6) and (7) don't hold otherwise}
\begin{subequations}\label{EQ:Specialized Dynamics}
	\begin{align}
	& \quad \; \; \;\dot z_{ij}=F_{ij}(\Zz)-\frac{1}{\No}\sum_{l=1}^{\No}F_{il}(\Zz) \label{eq:specialized a}\\
	&F_{ij}(\Zz)=-d z_{ij} + 
	u\left(  S_{1}\left(\alpha z_{ij} + \sum_{\substack{k \neq i \\k = 1 }}^\Na \gamma \tilde{a}_{ik} z_{kj}\right)+ \sum_{\substack{l\neq j\\ l=1}}^\No S_{2}\left( \beta z_{il} + \sum_{ \substack{k \neq i\\ k = 1}}^\Na \delta \tilde{a}_{ik} z_{kl} \right)\right)+b_{ij}.\label{eq:specialized b}
	\end{align}
\end{subequations}
In this regime, the intra-agent network is all-to-all and homogeneous, with weight $\alpha$ for same-option interactions and weight $\beta$ for inter-option interactions. The topology of the inter-agent  network is encoded in the adjacency matrix $\tilde A=[\tilde a_{ik}]\in\R^\Na\times\R^\Na$. For all existing inter-agent links, i.e., those satisfying $\tilde a_{ik}=1$, the inter-agent same-option weight is $\gamma$ and the inter-agent inter-option weight is $\delta$. It is natural to assume that
\begin{equation}\label{eq:self reinf mutual excl}
\alpha \geq 0,\quad\text{and}\quad \beta < 0,
\end{equation}
which means that opinions are self-reinforcing ($\alpha \geq 0$) and mutually exclusive ($\beta < 0$). The two remaining parameters, $\gamma$ and $\delta$, determine whether agents cooperate or compete in forming opinions. More precisely, if $\gamma>\delta$, agents {\it cooperate} in forming opinions, since the opinion of any agent $i$ about any option $j$ is more strongly excited by the opinions of its neighbors about the same option $j$ than by their opinions about all the other options $l\neq j$. 
Conversely, if $\gamma<\delta$, agents {\it compete} in forming opinions, since the opinion of any agent $i$ about any option $j$ is more strongly excited by the opinions of its neighbors about all the other options $l\neq j$ than by their opinions about the same option $j$. 

\af{
\begin{remark}
 Because intra-agent coupling is all-to-all and homogeneous in model~\eqref{EQ:Specialized Dynamics}, options are indistinguishable from a network topology perspective. A natural extension of model~\eqref{EQ:Specialized Dynamics} would be to consider a topologically richer intra-agent network. This generalization is natural in the presence of distinguished options or ``issues", i.e., subgroups of more tightly related options.
\end{remark}
}

For homogeneous networks, the cooperative or competitive nature of agent interactions is the key determinant of the opinion formation behavior. \af{The proof of the following theorem follows by linearizing the opinion dynamics~\eqref{EQ:Specialized Dynamics} at the neutral equilibrium and by studying the kernel of the resulting Jacobian matrix (see~\cite[Theorem 1]{bizyaeva2020patterns} for details). Let $\R\{v_1,\ldots,v_k\}$ be the span of vectors $v_1,\ldots,v_k\in\R^n$ and $\mathbf{1}_n^{\perp}$ be the orthogonal complement to the vector of ones in $\R^n$.
	
	\begin{theorem}\label{THM: dis-agree dicothomy}
		Suppose that $b_{ij}=0$ for all $i=1,\ldots,\Na$ and all $j=1,\ldots,\No$.
		\begin{itemize}
			\item If $\gamma>\delta$, then in model~\eqref{EQ:Specialized Dynamics} opinion formation emerges as a bifurcation from the neutral state for
			\begin{equation}\label{eq:agrement threshold}
			u=u_a:=\frac{d}{\alpha-\beta+\lambda_{max}(\gamma-\delta)}
			\end{equation}
			where $\lambda_{max}$ is the maximum eigenvalue of $\tilde A$. All bifurcation branches are tangent to the $(\No-1)$-dimensional kernel $\R\{v_{max}\}\otimes \mathbf{1}_{\No}^{\perp}$, where $v_{max}$ is the {\em positive} (Perron-Frobenius) eigenvector associated to $\lambda_{max}$. More precisely, each bifurcation branch is tangent to a one-dimensional subspace $\R\{v_{max}\otimes v_{ax}\}$ of $\R\{v_{max}\}\otimes \mathbf{1}_{\No}^{\perp}$, for some $v_{ax}\in\mathbf{1}_{\No}^{\perp}$. In particular, because every agent's opinion state is directly proportional (according to the entries of $v_{max}$) to the same opinion configuration vector $v_{ax}$, {\bf all the opinion-formation bifurcation branches are of agreement type}.
			
			\item If $\gamma<\delta$ and the minimum eigenvalue of $\tilde A$ is simple\footnote{The assumption that the minimum eigenvalue of $\tilde A$ has algebraic multiplicity one is usually true for all sufficiently large complex networks. The case in which this multiplicity is larger can also be analyzed, but for simplicity we do not address it here.}, then in model~\eqref{EQ:Specialized Dynamics} opinion formation emerges as a bifurcation from the neutral state for
			\begin{equation}\label{eq:disagrement threshold}
			u=u_d:=\frac{d}{\alpha-\beta+\lambda_{min}(\gamma-\delta)}
			\end{equation}
			where $\lambda_{min}$ is the minimum eigenvalue of $\tilde A$. All bifurcation branches are tangent to the $(\No-1)$-dimensional kernel $\R\{v_{min}\}\otimes \mathbf{1}_{\No}^{\perp}$,  where $v_{min}$ is the eigenvector associated to $\lambda_{min}$. Each bifurcation branch is tangent to a one-dimensional subspace $\R\{v_{min}\otimes v_{ax}\}$ of $\R\{v_{min}\}\otimes \mathbf{1}_{\No}^{\perp}$, for some $v_{ax}\in\mathbf{1}_{\No}^{\perp}$. %By orthogonality with $v_{max}$, 
			Because the vector $v_{min}$ is orthogonal to $v_{max}$, it has mixed-sign entries.  In particular, some agents have opinion state directly proportional to the opinion configuration vector $v_{ax}$, whereas some agents have opinion state  inversely proportional to it, and so {\bf all the opinion-formation bifurcation branches are of disagreement type}.
		\end{itemize}
	\end{theorem}
	
\begin{proof} Equations~\eqref{eq:agrement threshold} and~\eqref{eq:disagrement threshold} were derived in \cite{bizyaeva2020general}[Theorem IV.1] and \cite{bizyaeva2020patterns}[Theorem 1]. The form of the one-dimensional subspaces along which agreement and disagreement bifurcations occur, in particular, the vectors $v_{ax}$, can be constructed using equivariant bifurcation theory methods~\citep{Golubitsky1985-2,Golubitsky2002book}. Because the intra-agent opinion network is all-to-all and homogeneous, model~\eqref{EQ:Specialized Dynamics} is symmetric with respect to arbitrary option permutations, i.e., it is $\ES_\No$-equivariant, where $\ES_\No$ is the symmetric group of order $\No$. It follows from the Equivariant Branching Lemma~\citep[Therem XIII.3.3]{Golubitsky1985-2} that, generically, $v_{ax}$ is the generator of the fixed point subspace of one of the axial subgroups of the action of $\ES_\No$ on  $\mathbf{1}_{\No}^{\perp}$. There is, of course, only one (modulo multiplication by a non-zero scalar) such vector for $\No=2$, i.e., $v_{ax}=(1,-1)$. Modulo multiplication by a non-zero scalar, there are three such vectors for $\No=3$, i.e., $v_{ax}=(1,1,-2)$, $v_{ax}=(1,-2,1)$, $v_{ax}=(-2,1,1)$. The form of $v_{ax}$ for $N_o>3$ can be similarly derived.
\end{proof}

Although derived for the simplified homogeneous dynamics~\eqref{EQ:Specialized Dynamics}, Theorem~\ref{THM: dis-agree dicothomy} is predictive of the model behavior in the presence of small heterogeneities in the coupling parameter, i.e., when the homogeneous regime is weakly violated. Also, the same bifurcation techniques used to prove this theorem generalize to the case in which large heterogeneities are present. See, e.g., \cite[Section~III.C]{bizyaeva2020general} for a general clustering result in the presence of large difference in the coupling parameter between two or more agent subpopulations.

	\begin{remark}\label{RMK: swap sig and sum}
		In~\eqref{EQ:Specialized Dynamics}, we can bring the sums appearing inside $S_1$ and $S_2$ outside of the sigmoids without changing the qualitative behavior of the model. In particular, because the model Jacobian does not change if sums and sigmoids are swapped, Theorem~\ref{THM: dis-agree dicothomy} remains true.
	\end{remark}
}

\begin{figure}
	\centering
	\includegraphics[width=0.95\textwidth]{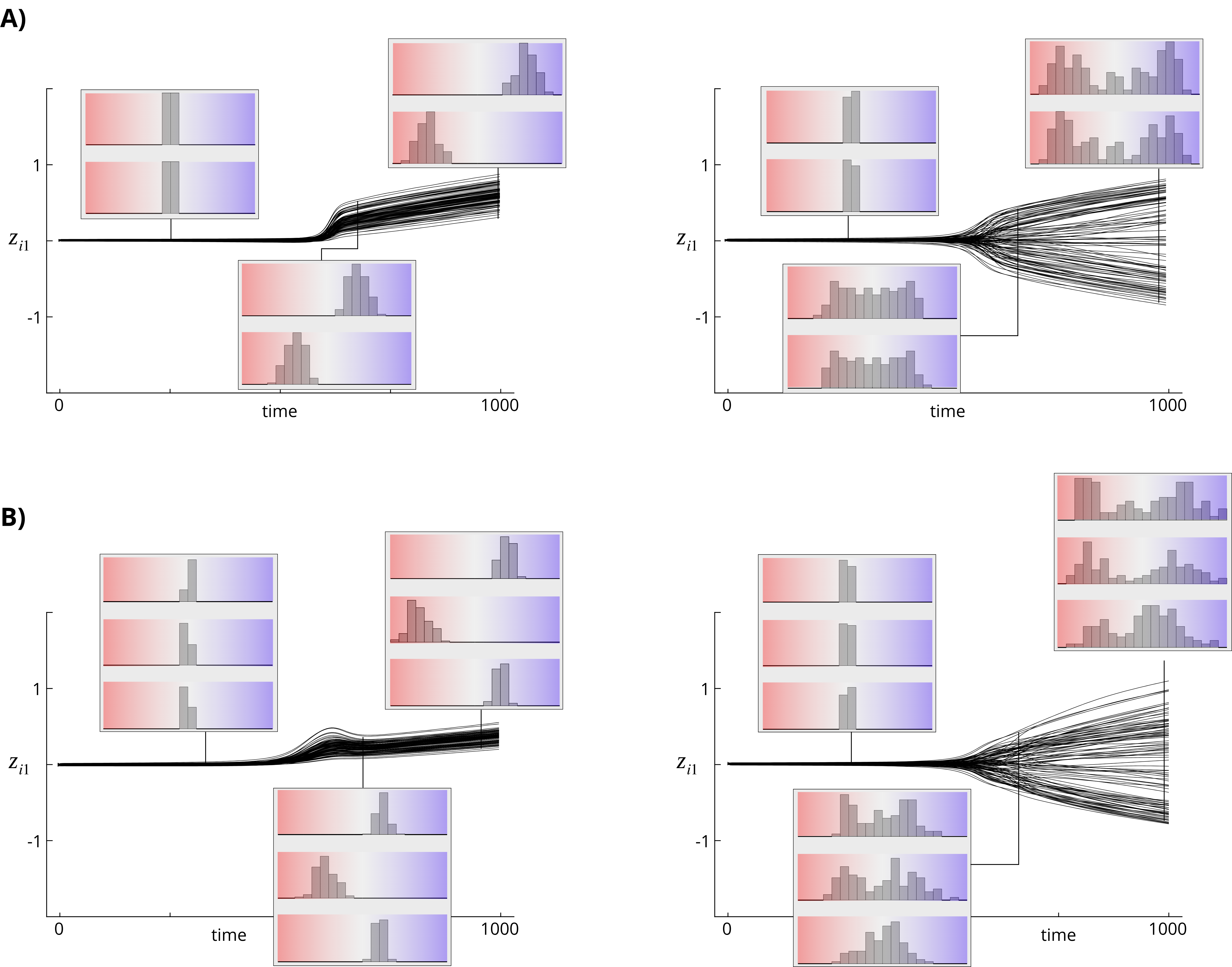}
	\caption{Agreement (left) and disagreement (right) opinion formation in Watts-Strogatz networks with $\Na=100$ agent nodes and (A) $\No=2$ and (B) $\No=3$  options. The Watts-Strogatz rewiring probability is $0.9$. \af{Other parameters: $\alpha=0.2$, $\beta=-0.5$, $\gamma=-\delta=0.1$ (agreement), $\gamma=-\delta=-0.1$ (disagreement). Initial conditions and inputs are drawn from a normal distribution with zero mean and variance $10^{-3}$. In all the simulations, $S_1(x)=\tanh(x+0.5\tanh(x^2))$ and $S_2(x)=0.5\tanh(2x+\tanh(x^2))$.}}
	\label{fig:dis-agreement WS}
\end{figure}

In Figures~\ref{fig:dis-agreement WS} and~\ref{fig:dis-agreement BA}, we illustrate the prediction of Theorem~\ref{THM: dis-agree dicothomy} for two types of (undirected) complex networks: a Watts-Strogatz  ``small-world'' network~\citep{watts1998collective} and a Barabasi-Albert ``scale-free'' network~\citep{albert2002statistical}, respectively. Panel A of both figures corresponds to $\No=2$ and panel B to $\No=3$. In the left plots of each panel, agents are cooperative, whereas in the right plots, agents are competitive. In all the simulations $\Na=100$ and small Gaussian-distributed random perturbations were added to all inputs $b_{ij}$ to verify the robustness of the theory. \af{Theorem~\ref{THM: dis-agree dicothomy} is indeed valid for $b_{ij}=0$ but, by normal hyperbolicity of the predicted bifurcation branches of agreement and disagreement equilibria, its predictions remain valid in the case where the $b_{ij}$ are not exactly zero or other types of perturbations (e.g., coupling weights perturbations that weakly violate the homogeneous regime assumption) are present. Precise robustness bounds can be computed using the methods introduced in~\cite{wang1994robustness}. This observation makes Theorem~\ref{THM: dis-agree dicothomy} relevant for real-world applications.} 

Each plot in Figures~\ref{fig:dis-agreement WS} and~\ref{fig:dis-agreement BA} shows the temporal evolution of $z_{i1}$, i.e.,  agent $i$'s opinion about the first option, as the attention parameter is slowly increased from below to above the agreement and disagreement opinion-formation critical values $u_a$ and $u_d$ defined in~\eqref{eq:agrement threshold} and~\eqref{eq:disagrement threshold}, respectively. The critical value is reached at half of the simulation time. The insets show the distribution of agent opinions about all the options at one third of the simulation time (i.e., before the opinion forming bifurcation), two thirds of the simulation time (i.e., right after the opinion-forming bifurcation), and at the end of the simulation time. The first (second, third) line shows the distribution of the agent opinions about the first (second, third) option. Colors indicate  the strength of favor/disfavor for a given option, with light blue (resp., red) indicating weak favor (resp. disfavor) and dark blue (resp., red) indicating strong favor (resp., disfavor). Asymmetric distributions on the right (left) indicate that agents agree in favoring (disfavoring) the option. A roughly symmetric distribution indicates disagreement about an option, with some agents favoring it and other agents disfavoring it. A sharply bimodal distribution corresponds to {\it polarization} of opinions about the option.

\begin{figure}
	\centering
	\includegraphics[width=0.95\textwidth]{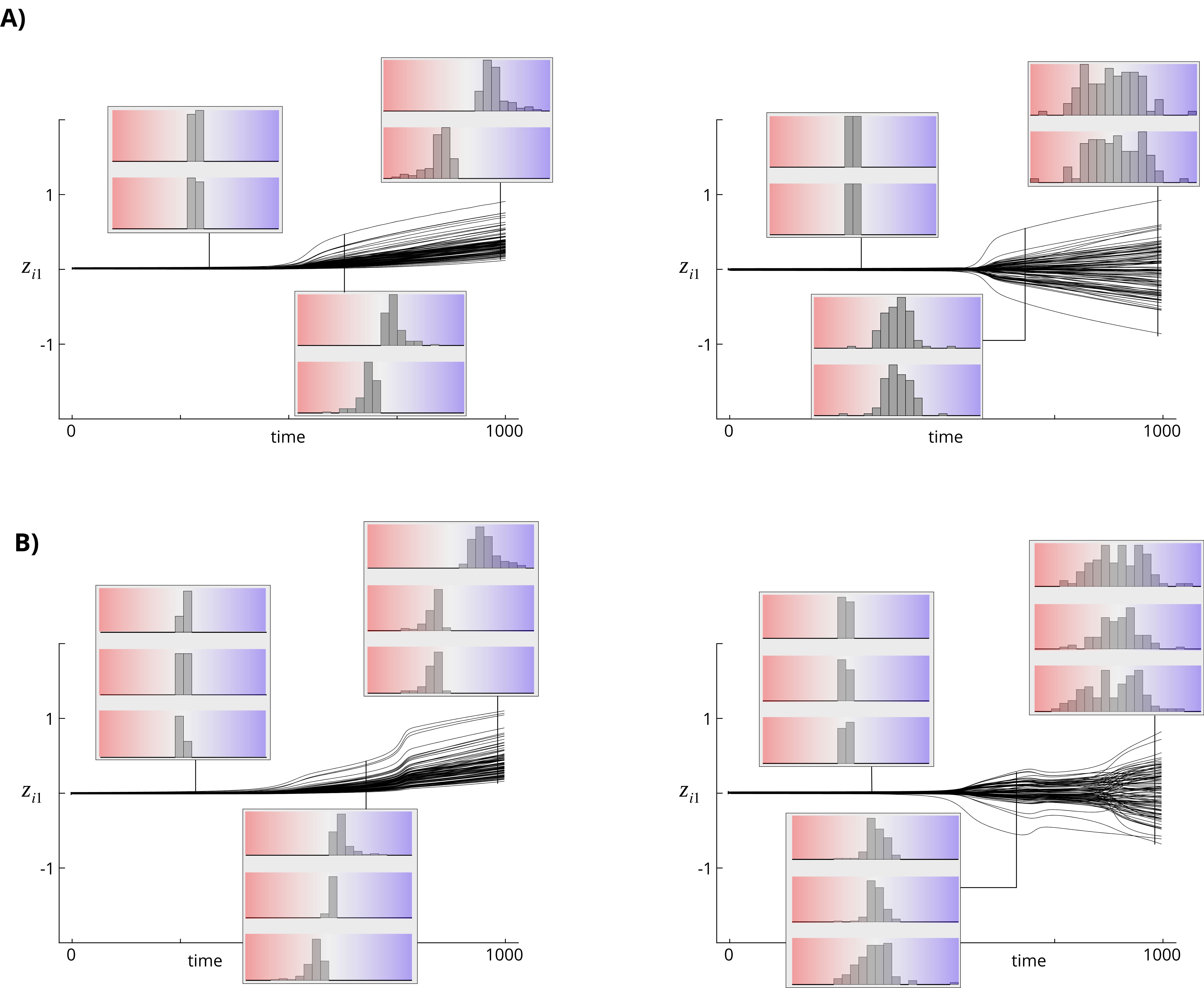}
	\caption{Agreement (left) and disagreement (right) opinion formation in Barabasi-Albert networks with $\Na=100$ agent nodes and (A) $\No=2$ options and (B) $\No=3$ options. The Watts-Strogatz rewiring probability is $0.9$. \af{Other parameters: $\alpha=0.2$, $\beta=-0.5$, $\gamma=-\delta=0.1$ (agreement), $\gamma=-\delta=-0.1$ (disagreement). Initial conditions and inputs are drawn from a normal distribution with zero mean and variance $10^{-3}$.}}
	\label{fig:dis-agreement BA}
\end{figure}

The model behavior follows our theoretical predictions. When agents cooperate, the neutral equilibrium bifurcates at $u=u_a$ into an agreement branch where all the agents share the same qualitative opinion. The larger is $u$ relative to $u_a$, the stronger are the emerging opinions. In Figures~\ref{fig:dis-agreement WS}A and~\ref{fig:dis-agreement BA}A, left, all the agents favor Option 1 and disfavor Option 2. In Figure~\ref{fig:dis-agreement WS}B, left, all agents disfavor Option 2 and are conflicted between Options 1 and 3. In Figure~\ref{fig:dis-agreement BA}B, left, all agents favor Option 1 and are conflicted between Options 2 and 3. Conversely, when agents compete, the bifurcation branches emerging from the neutral equilibrium are of disagreement type and again opinion strength increases as $u$ increases above $u_d$. In Figures~\ref{fig:dis-agreement WS}A and~\ref{fig:dis-agreement BA}A, right, roughly half of the agents favor Option 1, while the other half favor Option 2. In Figures~\ref{fig:dis-agreement WS}B and~\ref{fig:dis-agreement BA}B, right, for each of the three options the group is roughly split into agents that favor the option and agents that disfavor it.

\af{
\begin{remark}
 The precise effect of the inputs $b_{ij}$ on the bifurcation and steady-state behavior can be computed using Lyapunov-Schmidt reduction techniques~\citep{Golubitsky1985}. More precisely, the projection of the input vector onto the bifurcation kernel predicts, to first order, in which direction along the one-dimensional kernel's subspaces, $\mathbb R\{v_{max}\otimes v_{ax}\}$ (agreement case) or $\mathbb R\{v_{min}\otimes v_{ax}\}$ (disagreement case) defined in Theorem~\ref{THM: dis-agree dicothomy},  the agreement or disagreement bifurcation unfolds (see~\cite[Theorem~1]{Gray2018} for the agreement, 2-option case). At bifurcation, the unfolding direction determines which options are favored or disfavored by the inputs (see, e.g., \citep[Figure~2]{Gray2018}). Here we generalize ~\citep[Theorem~1]{Gray2018} to the disagreement and 3-option cases.
\end{remark}
}

\section{\af{Multi-Robot Task Allocation as Disagreement Opinion Formation}}
\label{sec:task allocation}

%We present two examples to demonstrate how the opinion formation model \eqref{EQ:generic decision dynamics} can be applied in swarm robot control and coordination. Though our analyses discussed in Sections~2 and 4 focus on the constant input case, i.e., $b_{ij}$ is constant, in both of the examples, we consider feedback design of the input $b_{ij}$ to highlight the flexibility and applicability of our model \eqref{EQ:generic decision dynamics} in designing distributed decision-making algorithms for swarm robot applications.

%\subsection{Multi-Robot Task Allocation}

We illustrate using a detailed example how our opinion dynamics~\eqref{EQ:Specialized Dynamics} can be applied in robot swarm coordination and control, in particular, to a distributed task-allocation problem. Given its dynamical nature and the small number of parameters, the proposed distributed control algorithm can be implemented in analog circuits and is thus a good candidate for low-computational-power robot swarms that require dynamic task allocation. A single-agent motivation dynamics framework was also recently proposed and applied to a task coordination problem for a mobile robot \citep{Reverdy2018,Reverdy2020}. In the spirit of these works, the distributed opinion dynamics of \cite{bizyaeva2020general} allows us to move beyond a single robotic unit and consider distributed task coordination in a robotic swarm.   

Consider that a robot swarm is given $\No$ tasks to complete, in which each task has a distinct priority level and each robot may have its own preference to work on certain tasks. To formalize, we denote by a positive constant $\mu_j \in [0, 1]$, $j=1,\ldots,\No$, satisfying $\sum_{j=1}^\No \mu_j = 1$, the priority of each task~$j$. The main objective in this example is to design a simple distributed decision-making algorithm that assigns robots to each task. The algorithm should be designed to respect task priority, i.e., the larger the task priority $\mu_j$, the larger the number of robots assigned to that task by the algorithm.

Given an undirected communication network between the robots, with adjacency matrix $[\tilde a_{ik}]_{i,k=1}^\Na$, let $\mathcal N_i$ be set of neighbors of robot $i$. We adapt opinion dynamics~\eqref{EQ:Specialized Dynamics} to the task-allocation problem as follows. The preference that robot $i$ gives to performing task $j$ evolves accordingly to vector field
\begin{subequations} \label{eq:multirobot_task_allocation}
	\begin{align}
	\dot z_{ij} &= F_{ij}(\mathbf Z) - \frac{1}{N_o} \sum_{l=1}^{N_o} F_{il}(\mathbf Z) \\
	F_{ij}(\mathbf Z) &= -z_{ij}  + u \, \left(\mu_j \, (|\mathcal N_i|+\nu_{ij}) - \frac{1}{2} \sum_{k \in \mathcal N_i} S \left( 2\tilde\gamma z_{kj} \right) \right). \label{eq:multirobot_task_allocation_b}
	\end{align}
\end{subequations}
The positive parameter $\nu_{ij}$ is the intrinsic zealousness of robot $i$ to perform task $j$. As we will see, the zealousness parameter can be used as control input to the robot swarm to trigger the distributed task allocation. Opinion dynamics~\eqref{EQ:Specialized Dynamics} specialize to model~\eqref{eq:multirobot_task_allocation} by letting $b_{ij}=u\mu_j(|\mathcal N_i|+\nu_{ij})$, $\alpha=\beta=\delta=0$, $\tilde\gamma=-\gamma>0$, and by swapping the inter-agent interaction sum with the sigmoid $S_1$ (as discussed in Remark~\ref{RMK: swap sig and sum}, this choice does not affect the qualitative behavior of the model).

The rationale behind the derivation of model~\eqref{eq:multirobot_task_allocation} is the following. Let $N_i^j$ be the number of robots in $\mathcal N_i$ that express preference for task $j$, i.e., such that $z_{kj}>0$.  Because generically $z_{kj}\neq 0$ and if the upper and lower  limits of $S$ are $\pm1$, respectively, for sufficiently large $\gamma$ the term $\frac{1}{2} \sum_{k \in \mathcal N_i} S \left( 2\gamma z_{kj} \right)\approx \frac{1}{2} \left( N_i^j - \left( |\mathcal N_i| - N_i^j \right) \right) =  N_i^j - \frac{|\mathcal N_i|}{2}$. Imposing the simplex condition by subtracting the average opinion drift as in~(\ref{eq:multirobot_task_allocation}a), the constant $\frac{|\mathcal N_i|}{2}$ disappears from the opinion dynamics on the simplex and, for $\nu_{ij}=0$ the effective interaction term reduces to
\begin{align}
\mu_j \, |\mathcal N_i| - \frac{1}{2}\sum_{k \in \mathcal N_i} S \left( \epsilon z_{kj} \right)  \approx \mu_j \, |\mathcal N_i| - N_i^j\,.
\end{align}
In other words, in our task-allocation model, each robot updates its preference for task $j$ roughly proportionally to the priority of task $j$, scaled by the number of its neighbors, minus the number of its neighbors that already show a preference for the same task. A non-zero zealousness $\nu_{ij}$ increases the tendency of robot $i$ to choose task $j$.

Observe that the condition $-\tilde\gamma=\gamma<0$ and $\delta=0$ imposes that the model is in the disagreement regime for $u>u_d>0$. Thus, in our model, the robots naturally develop different opinions about the tasks, i.e., they self-allocate across the task in a distributed fashion. Furthermore, robot task-allocation will reflect task priorities, which, from a bifurcation viewpoint, breaks the symmetry between the tasks.

\begin{figure}
	\centering
	\includegraphics[width=0.75\textwidth]{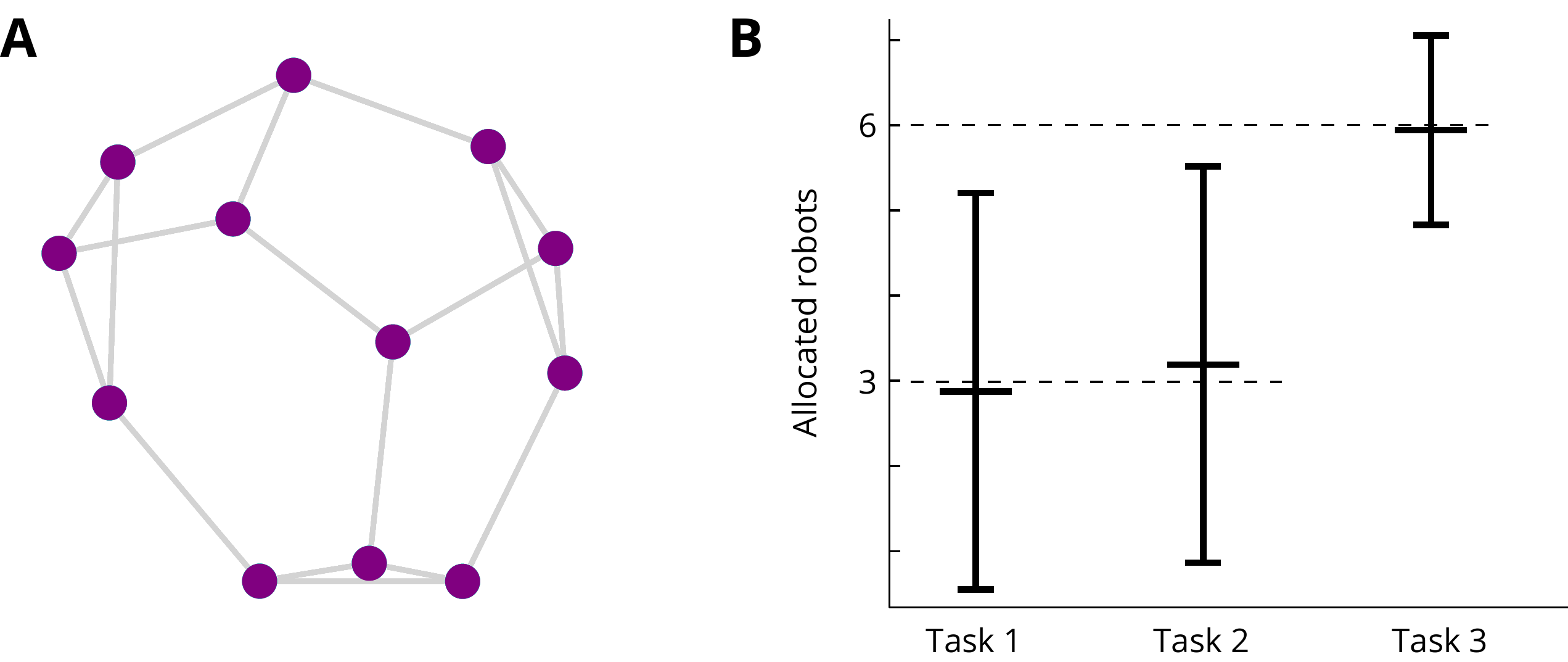}
	\caption{A. The Frucht Graph. B. Average (large bar) and standard deviation (short bar) of the number of agents allocated to Tasks~1,2,3 over one-hundred instances of our task-allocation algorithm. Details in the main text.}
	\label{fig:task allocation}
\end{figure}

We illustrate the performance of our task-allocation algorithm for a swarm of $\Na=12$ robots allocating themselves across $\No=3$ tasks. We assume that the inter-agent communication graph is the Frucht graph~\citep{frucht1939} (see Figure~\ref{fig:task allocation}A), which is 3-regular ($|\mathcal N_i|=3$ for all $i=1,\ldots,\Na$) but possesses no non-trivial symmetries, i.e., all robots can be distinguished by their position in the graph. Our algorithm is particularly suited in {\it a priori} deadlocked situations, i.e., those in which two or more tasks have the same objective priority. In our example, we set $\mu_1=\mu_2=0.3$ and $\mu_3=0.4$. We also set $\tilde\gamma=1.0$ and $u=2    u_d$. For the sigmoid function we use $S(x)=\tanh(x+0.5\tanh(x^2))$, where the $\tanh(x^2)$ ensure genericity (i.e., non-vanishing even partial derivatives) of the $S$ at the origin. We run one hundred repetitions of our example with small random zealousness parameters $\nu_{ij}$ (modeling parametric uncertainties) and small random initial conditions, and computed the resulting average and standard deviation of the number of robots allocated to each option. 

Robot $i$ is considered allocated to task $j$ if at the end of the simulation $\|\Zz_i\|\geq\theta$ and $z_{ij}>z_{il}$ for all $l\neq j$, where $\theta$ is a control threshold parameter that determines whether a robot became opinionated (in the simulations, $\theta=0.1$). The final simulation time can be considered as a tunable decision deadline for the task-allocation decision. Figure~\ref{fig:task allocation}B shows the results of our analysis. On average, a larger number of agents is allocated to the highest-priority task. A roughly equal number of agents is allocated on average to the less urgent tasks. The standard deviation of the allocation number is also smaller for the highest-priority task, which indicates that our algorithm consistently allocates the largest number of robots to the highest-priority task. Of course, there is a lot of room for improving the simple algorithm proposed in this section, in particular, in order to drastically reduce variability  in the number of agents allocated to the various tasks, and thus increase robustness to initial conditions and parameter uncertainties. In some instances of our example, depending on initial conditions or the intrinsic zealousness parameters, no robots were allocated to either Task~1 or Task~2, which is clearly an undesirable situation in practice. We suggest a strategy to overcome these limitations in Section~\ref{sec:task allocation cascades}.

\section{Opinion cascades with tunable sensitivity from feedback attention dynamics}

\label{sec: opinion cascades}

In this section we couple the opinion dynamics~\eqref{EQ:Specialized Dynamics} with the attention dynamics introduced in~ \cite{bizyaeva2020general}. Using a detailed geometrical analysis for the case $\No=3$ (the case $\No=2$ is studied in~\cite{bizyaeva2020general}), we show that the resulting positive feedback between attention and opinion strength allows an agent to transition from being in a stable unopinionated state to a stable opinionated state in response to an input $\mathbf{b}_{i} = (b_{i1}, \ldots, b_{iN_o})$ sufficiently large in magnitude. The associated activation threshold is implicitly defined by the geometric properties of the dynamics and is explicitly tunable with the model parameters. The threshold effect mediated by the positive feedback between attention and opinion strength scales up to networks of interconnected agents under the influence of inputs on a small number of ``seed nodes". At the network level, the coupling between opinion and attention dynamics creates a threshold for the input magnitude below which the network remains close to neutral and above which an opinion cascade is triggered.
%We show that the propensity of a network to develop large cascade  is predicted by its agreement and disagreement centrality distribution.

\subsection{Feedback attention dynamics and tunable sensitivity of opinion formation at the single agent level}\label{sec:attention dynamics}

%So far we have discussed the framework for the dynamics of agent opinions \eqref{EQ:Specialized Dynamics}, and how opinionated agreement and disagreement group equilibria appear as bifurcations from the neutral state $\Zz = 0$ with the attention parameter $u_{i} = u$ as the relevant bifurcation parameter. 

To understand the effect of coupled opinion and attention dynamics we first consider an agent without any neighbors. Even without an inter-agent communication network, the dynamics of each agent's opinions are organized by a bifurcation at which the agent develops preferences for the available options. When $\tilde a_{ik} = 0$ for all $i,k = 1, \dots, N_{a}$, the $F_{ij}(\Zz)$ function in  \eqref{eq:specialized b} becomes 
\begin{equation}
F_{ij}(\Zz) = -d z_{ij} + 
u_{i}\left(  S_{1}\left(\alpha z_{ij} \right)+ \sum_{\substack{l\neq j\\ l=1}}^\No S_{2}\left( \beta z_{il} \right)\right)+b_{ij}. \label{eq:F_decoupled}
\end{equation}
When $b_{ij} = 0$, for all $j = 1, \dots, N_{o}$, and $u_{i} < \frac{d}{\alpha - \beta}$, the neutral opinion state $\Zz_{i} = \mathbf{0}$ of agent $i$ is stable. At $u_{i}^{*} = \frac{d}{\alpha - \beta}$ branches of opinionated solutions emerge in an $S_{N_{o}}$-equivariant bifurcation, which is studied in detail in \cite{elmhirst2004s}. We omit the formal details of this bifurcation analysis here, and simply note that for $u_{i} < u_{i}^{*}$ the agent's opinion converges to the neutral equilibrium with $\| \Zz_{i}^{*} \| = 0$, whereas for $u_{i} > u_{i}^{*}$ it converges to some nonzero magnitude equilibrium $\| \Zz_{i}^{*} \| > 0$, whose norm grows monotonically with $u_i$. The stable opinionated equilibria that appear as primary bifurcation branches correspond either to agent $i$ choosing one option and rejecting all others, or to it being conflicted between a subset of the available options (Figure~\ref{fig:SensitivityHysteresis}A, left).

The dynamics of an agent $i$'s opinion with nonzero input $\mathbf{b}_{i}$ directly follow from the bifurcation behavior of its opinion dynamics with $\mathbf{b}_{i}=\mathbf{0}$. As discussed in Section \ref{SEC: general model}, when $u_{i}$ is small, the input  $\mathbf{b}_{i}$ and resistance $d$ define the linear opinion formation behavior of agent $i$. For large $u_{i}$ the nonlinear behavior dominates and agent $i$'s opinion can become much larger in magnitude than its input $\mathbf{b}_{i}$. Generically, when $\mathbf{b}_{i} \neq 0$, for values of $u_{i}$ near $u_{i}^{*}$, a single branch of stable opinionated solutions is selected accordingly to ${\mathbf b}_i$, while all other solutions are removed (Figure~\ref{fig:SensitivityHysteresis}A, right). 

How the opinion bifurcation diagram changes in response to inputs motivates the design of an attention feedback dynamic that increases an agent's attention from below to above its bifurcation point when sufficiently large inputs are present. We define the closed-loop attention dynamics as
\begin{equation}
\tau_{u} \Dot{u}_{i} = - u + u_{min} + (u_{max}-u_{min}) S_{u} (\| \Zz_{i} \|)  \label{eq:udot}
\end{equation}
where $\tau_{u}$ is an integration timescale and $S_{u}:\R_{\geq 0}\to[0\ 1]$ is a smooth saturating function, satisfying $S_{u}(0) = 0$, $S_{u}(y) \to 1$ as $y \to \infty$, $S_{u}'(y) > 0$ for all $y \in \R_{\geq 0}$, $S''_{u}(y) > 0$ for all $0\leq y < u_{th}$ and $S''_{u}(y) < 0$ for all $y > u_{th}$, with $u_{th} >0$. For the sake of numerical illustrations we use $S_u(y)=\frac{y^n}{u_{th}^n+y^n}$. The larger is $n$ the steeper is the sigmoid around its threshold $u_{th}$. We assume $u_{min} < u_{i}^{*} < u_{max}$ and $n$  sufficiently large so that the saturation is sufficiently steep.

%Stable and unstable equilibria of the coupled system \eqref{eq:F_decoupled},\eqref{eq:udot} are points at the intersection of the nullcline surfaces in $V \times \mathds{R}$ defined by $\dot{\Zz}_i = 0$ and $\dot{u}_i = 0$. The points $u_i^s$ on the $u_{i}$-nullcline are given by 
%\begin{equation}
%    u_{i}^{s} = u_{min} + (u_{max}-u_{min}) S_{u} (\| \Zz_{i} \|).  \label{eq:u-nullcline}
%\end{equation}
%The sigmoid $S_{u}$ creates an implicit threshold which can be tuned (i.e. increased or decreased) with the single parameter $u_{th}$. For $\|\Zz_i\|$ below this threshold, 
%the equilibrium value of attention
%$u_i^s$ in \eqref{eq:u-nullcline} remains small, and for larger $\| \Zz_i \| > u_{th}$, $u_i^s$ saturates to a value near $u_{max}$. This implicit threshold built in to the agent's attention dynamics gives rise to the notion of \textit{tunable sensitivity} of the agent's coupled opinion and attention dynamics to input $\mathbf{b}_i$.

\begin{figure}
	\centering
	\includegraphics[width=0.7\textwidth]{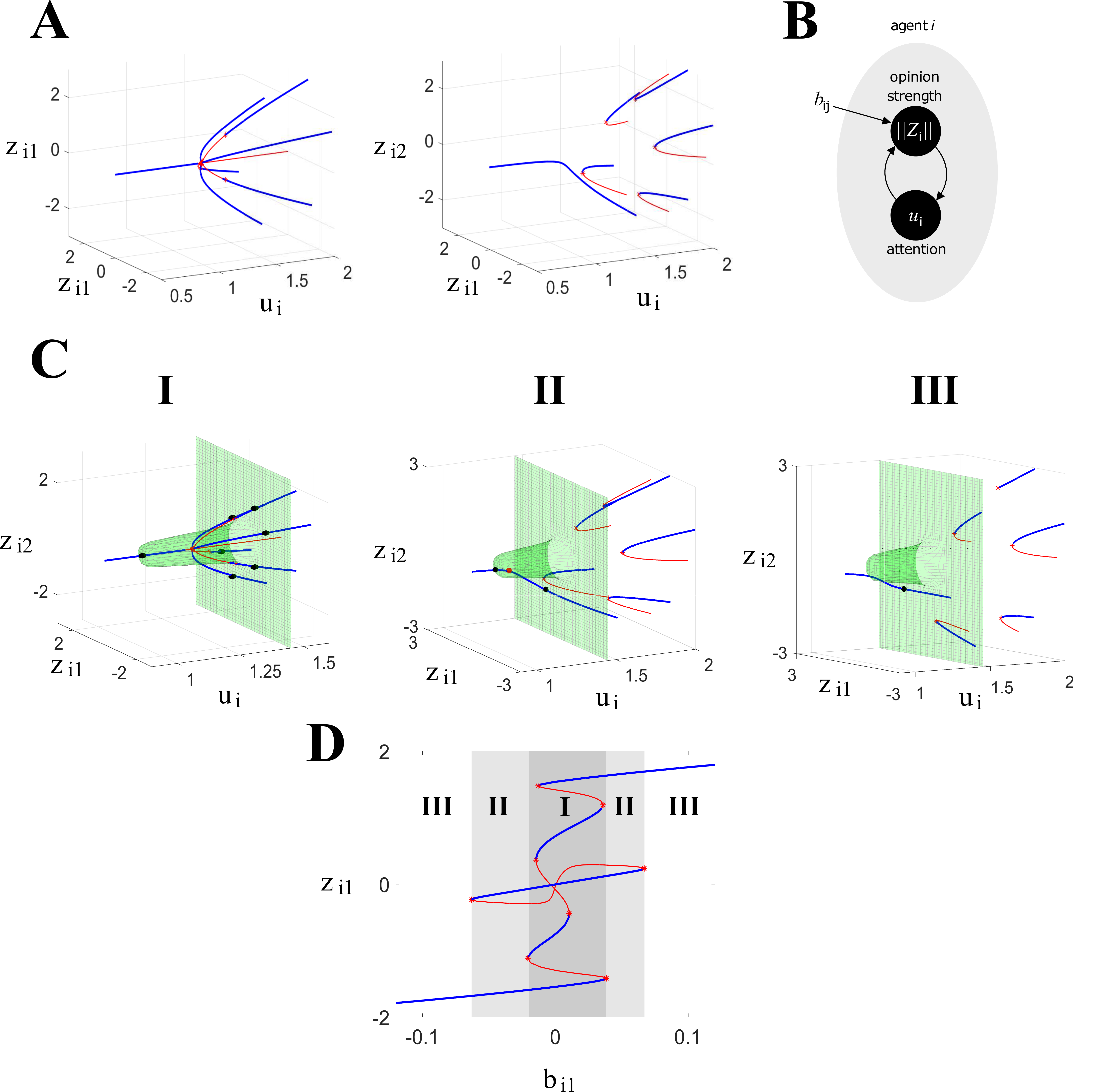}
	\caption{A. Bifurcation diagram of the single agent dynamics with respect to the attention parameter. Left: ${\mathbf b}_i=0$. Right: ${\mathbf b}_i\neq 0$. B. Schematic of the relationship between an agent's opinion strength, attention, and input.  C. Representative intersections of the attention nullclines (green surface) and opinion nullclines (red and blue lines) for the three parameter regimes with $N_{o} = 3$. D. Bifurcation diagram of the coupled system \eqref{eq:F_decoupled},\eqref{eq:udot} with $N_{o} = 3$, projected to the $b_{i1}$, $z_{i1}$ plane. In D blue (red) lines correspond to stable (unstable) equilibria of the coupled system. In A,C blue (red) lines correspond to stable (unstable) equilibria of the opinion dynamics \eqref{eq:F_decoupled} with a static attention parameter. Additional branches of unstable equilibria are not shown in A,C-I for diagram clarity. Black (red) spheres mark the stable (unstable) equilibria of the coupled system.  Parameters: $d = 1$, $\alpha  = 0.4$, $\beta = -0.4$, $u_{min} = 1.05$, $u_{max} = 1.45$, $n = 3$, $u_{th} = 0.15$. B: $b_{i2} = 0.005, b_{i3} = -0.002$, $b_{i1}$ varied as illustratd. A (left) and CI: $\mathbf{b}_i = 0$; A (right) and CII : $\mathbf{b}_{i} = ( 0.05,-0.01,0.02 )$; CIII: $\mathbf{b}_{i} = ( 0.1,-0.01,0.02 )$.  Diagrams in A,C,D are generated with help of MatCont numerical continuation package \citep{matcont}.}
	\label{fig:SensitivityHysteresis}
\end{figure}

The closed-loop attention dynamics~\eqref{eq:udot} create a localized positive feedback between an agent's attention and its opinion strength, as illustrated in Fig. \ref{fig:SensitivityHysteresis}B. We use the chain rule to compute partial derivatives at $\Zz_i = 0$ of the time-derivative of $z_{ij}^2$ with respect to attention $u_{i}$ of an agent whose dynamics are defined by \eqref{eq:udot} and, conversely, the time-derivative of $u_i$ \eqref{eq:udot} with respect to $z_{ij}^2$. The sign of the product of these two partial derivatives gives the sign of the feedback between attention and opinion strength at the neutral equilibrium. Computing, we get
\begin{align*}
\frac{\partial\dot{(z_{ij}^2)}}{\partial u_{i}}&= 2 z_{ij} \left( \frac{N_o - 1}{N_o}\Big(S_1(\alpha z_{ij}) -  S_2(\beta z_{ij}) \Big) - \frac{1}{N_o} \sum_{\substack{l \neq j \\ l = 1}}^{N_o} \Big(S_1(\alpha z_{il}) -  S_2(\beta z_{il}) \Big) \right),\\
\frac{\partial \dot{u}_{i}}{\partial(z_{ij}^2)} &= \frac{1}{\tau_u}(u_{max} - u_{min}) S_u'(\| \Zz \|).
\end{align*}
Near the neutral equilibrium $\Zz_i = 0$ we can approximate $S_1(\alpha z_{il}) -  S_2(\beta z_{il}) \approx (\alpha - \beta) z_{il}$, and using the simplex constraint on the opinions we find that for small $\Zz_i$ and any $j = 1, \dots, N_{o}$,  $\frac{\partial\dot{(z_{ij}^2)}}{\partial u_{i}} \approx 2 (\alpha - \beta) z_{ij}^{2}$. Therefore
\begin{equation}
\frac{\partial\dot{(z_{ij}^2)}}{\partial u_{i}}\frac{\partial \dot{u}_{i}}{\partial(z_{ij}^2)}  \approx \frac{2}{\tau_u}(\alpha - \beta) (u_{max} - u_{min})  z_{ij}^{2} S_u'(\| \Zz \|). \label{eq:uz_loop}
\end{equation}
The right-hand side of \eqref{eq:uz_loop} is positive whenever $z_{ij} \neq 0$. Thus, a positive feedback loop between an agent's opinion strength and its attention  is triggered as soon as its opinion deviates from zero.
%How far the attention increases from small initial conditions and nonzero bias can be tuned with parameters of $S_{u}$ in \eqref{eq:udot}.  This tunable sensitivity can be understood in terms of three parameter regimes, illustrated geometrically for $N_o = 3$ in Fig. \ref{fig:SensitivityHysteresis}C,D. Define $\mathbf{b}_{i}^{\perp}$ to be a projection of input $\mathbf{b}_{i}$ onto $V_{a}$.

The localized positive feedback between attention and opinion strength creates multistability and threshold behavior in %the input-output characteristic of 
the closed-loop opinion dynamics. Figures~\ref{fig:SensitivityHysteresis}C and~\ref{fig:SensitivityHysteresis}D illustrate the three-dimensional phase portrait and the steady-state behavior (bifurcation diagram) of the closed-loop opinion-attention dynamics for $\No=3$, respectively. The bifurcation diagram in  Fig.~\ref{fig:SensitivityHysteresis}D is projected onto the $z_{i1}$,$b_{i1}$ plane.  Note that the bifurcation diagram characterizes the input-output relationship, where the input  is ${\mathbf b}_i$ and the output is the steady-state opinion. 

We can distinguish three regimes (labels as in Fig.~\ref{fig:SensitivityHysteresis}D):
\begin{enumerate}[label=\Roman*.]
	\item \textbf{Multistability:} $\|\mathbf{b}_{i}^{\perp}\| < \varepsilon_1$. For small  $\|\mathbf{b}_{i}^{\perp}\|$, the opinion and attention nullclines intersect in many locations and a number of these intersections are stable equilibria. One of the intersections is near to $u_{i} = u_{min}$ and $\Zz_i = \mathbf{0}$ and is stable. The other stable equilibria correspond to intersections of opinionated solution branches with attention nullcline near $u_i = u_{max}$.  For small initial conditions,  agent $i$'s opinion and attention remain small, attracted by the stable equilibrium near to $u_{i} = u_{min}$. For large initial conditions they  approach one of the opinionated equilibria. In this parameter regime the initial conditions determine whether agent $i$ becomes opinionated, and also which opinion it forms. 
	\item \textbf{Bistability:} $\varepsilon_1 < \|\mathbf{b}_{i}^{\perp}\|  < \varepsilon_2$.  For intermediate  $\|\mathbf{b}_{i}^{\perp}\|$, the coupled dynamics admit only  two stable equilibria: one is the unopinionated equilibrium near $u_i = u_{min}$ and the other is an opinionated equilibrium on the solution branch favored by the input $\mathbf{b}_i$. For small initial conditions, agent $i$'s opinion and attention remain small, wherease for large initial conditions they approach the stable equilibrium picked out by the input. In this regime, initial conditions determine whether or not an agent becomes opinionated, but the input governs the opinion formation outcome. 
	\item \textbf{Cascade:} $\| \mathbf{b}_{i}^{\perp} \| > \varepsilon_2$.  For sufficiently large  $\|\mathbf{b}_{i}^{\perp}\|$, the coupled system has a unique stable equilibrium point which corresponds to an opinionated state favoring the opinion picked out by the input $\mathbf{b}_i$. From any initial condition, an agent will reliably become opinionated in the direction informed by the input. 
\end{enumerate}
The boundaries $\varepsilon_1$, $\varepsilon_2$ between these three regimes can be tuned by the parameters of the saturation function $S_{u}$ in \eqref{eq:udot}. In other words, the sensitivity of the closed-loop dynamics to forming opinions in response to inputs is tunable.

\subsection{Network opinion cascades with tunable sensitivity}
\label{ssec:cascade tuneable}

In the presence of symmetric network interactions, the intra-agent positive feedback loop between attention and opinion strength scales up to the network level, as illustrated in Figure~\ref{fig:attention loop net}. Independent of their excitatory or inhibitory nature, network interactions create a positive feedback loop between the opinion strengths of any pair of interconnected agents, at least close to the neutral equilibrium where the interaction nonlinearities are steeper and thus the opinion interaction is stronger. To show this statement we use the chain rule to compute partial derivatives at the neutral equilibrium of the time-derivative of $z_{ij}^2$ with respect to $z_{kl}^2$ for two agents $i$ and $k$ interacting according to~\eqref{EQ:Specialized Dynamics} and for symmetric network coupling. Computing, we get
\begin{equation*}
\left.\frac{\partial\dot{(z_{ij}^2)}}{\partial z_{kl}^2}\right|_{\Zz=0}=\frac{1}{\No}\frac{z_{ij}}{z_{kl}}u_i\tilde a_{ik} (\delta-\gamma),\text{ if }l\neq j,\quad  \left.\frac{\partial\dot{(z_{ij}^2)}}{\partial z_{kj}^2}\right|_{\Zz=0}=\frac{\No-1}{\No}\frac{z_{ij}}{z_{kj}}u_i\tilde a_{ik} (\gamma-\delta)
\end{equation*}
%where $p=1$ if $j=l$, $p=2$ if $j\neq l$, and $\star_i$ denotes the argument of the sigmoid $S_p$ defined as in~\eqref{EQ:generic decision dynamics} in the current network state.
Thus, if $\tilde a_{ik}=1$, by symmetry,
\begin{equation}
\left.\frac{\partial\dot{(z_{ij}^2)}}{\partial z_{kl}^2} \frac{\partial\dot{(z_{kl}^2)}}{\partial z_{ij}^2}\right|_{\Zz=0}=\frac{1}{\No^2}u_iu_k(\gamma-\delta)^2,\text{ if }l\neq j,\quad \left.\frac{\partial\dot{(z_{ij}^2)}}{\partial z_{kj}^2} \frac{\partial\dot{(z_{kj}^2)}}{\partial z_{ij}^2}\right|_{\Zz=0}=\left(\frac{\No-1}{\No}\right)^2u_iu_k(\gamma-\delta)^2,
%u_iu_kS'_p(\star_i)S'_p(\star_k)(A^{jl}_{ik})^2>0,
\end{equation}
which yields an element-wise positive feedback loop between the opinion strengths of agent $i$ and agent $k$. %\ab{Doesn't (11) hold only if $A_{jl}^{ik} =  A_{lj}^{ik}?$ Also I think you have agent and option indices flipped here - $ik$ should be at the bottom and $jl$ should be at the top }
As a consequence, when equations~\eqref{EQ:generic decision dynamics} and~\eqref{eq:udot} are coupled, the attention variables of agents $i$ and $k$ are also interconnected in an (indirect) positive feedback loop. This can easily be seen by computing the overall sign of the interconnection
loop from $u_i$ to $u_k$ and back, passing through $||\Zz_i||$ and $||\Zz_k||$ (see Fig.~\ref{fig:attention loop net}). The existence of a networked positive feedback between the agents' attention variables and opinion strengths suggests the existence of a {\it network threshold for the trigger of opinion cascades}.

\begin{figure}
	\centering
	\includegraphics[width=0.5\textwidth]{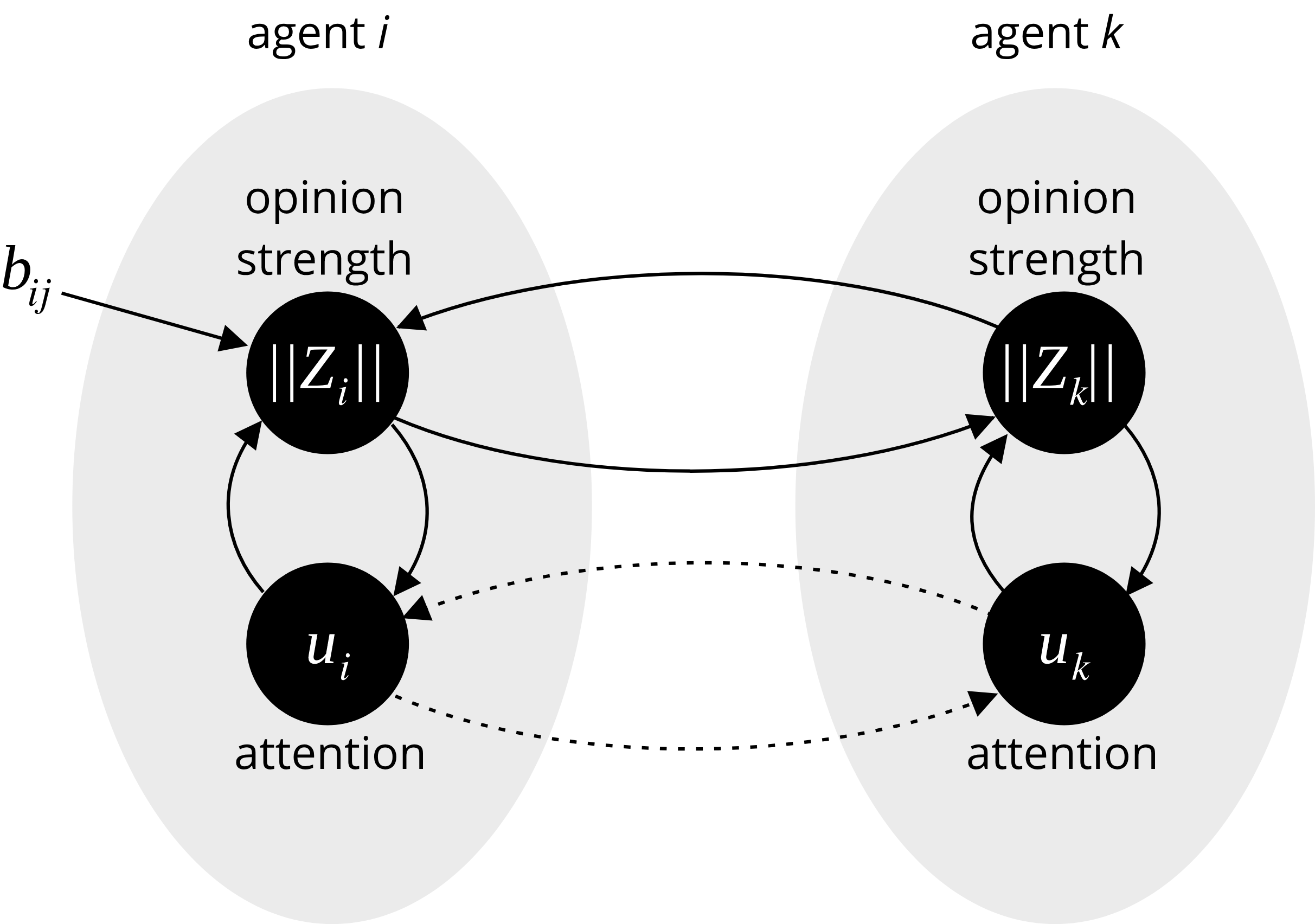}
	\caption{Intra- and inter-agent positive feedback loops between the agents' opinion strengths and their attention.}
	\label{fig:attention loop net}
\end{figure}

An {\it opinion cascade} refers to the situation in which only a small number $N_{seeds}\ll\Na$  of agents, called {\em seeds}, receive non-zero inputs $\bs{b}_1,\ldots,\bs{b}_{N_{seeds}}$ and yet a large part of the network develops a markedly non-zero opinion. %We call the $N_{seeds}$ agents receiving non-zero inputs the {\it seeds} of the opinion cascade. 
In the network example of Figure~\ref{fig:attention threshold}A, the cascade seeds are indicated with blue \af{and red arrows depending on the options favored or disfavored by the input (see the figure caption for details)}. When the inputs to the seeds are small, only the seeds develop clear opinions \af{aligned with their respective inputs}, while the other agents remain weakly opinionated or unopinionated in such a way that the network average opinion strength remains small (Figure~\ref{fig:attention threshold}B1). When the inputs to the seeds are sufficiently strong, the seeds develop opinions that are sufficiently strong to engage the positive feedback of the network attention, which spreads to some or all of the rest of the agents with some variable delay. %. For large times, 
This yields an opinion cascade such that the network average opinion strength grows large (Figure~\ref{fig:attention threshold}B2). %At the end of the simulation time, a large of the agents developed a strong opinion, despite only five of them received inputs. An opinion cascade. 
Whether the cascade is towards  agreement or disagreement depends on the cooperative (left plots) or competitive (right plots) nature of the agents. Observe that in both agreement and disagreement cases, a number of agents have afterthoughts: they develop an early opinion about Option~1 but either exhibit the opposite opinion or become neutral by the end of the cascade process.

\begin{figure}
	\centering
	\includegraphics[width=\textwidth]{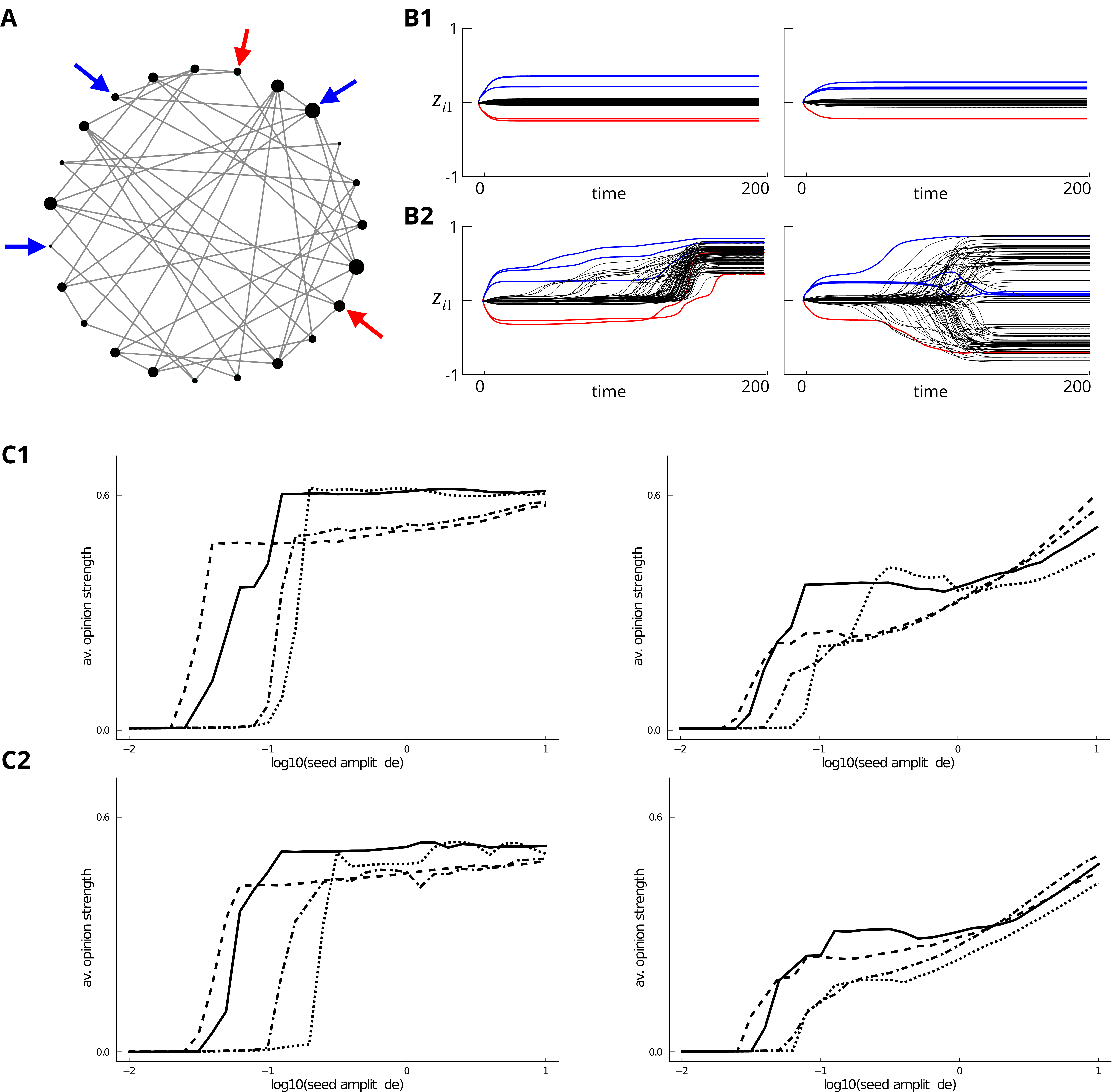}
	\caption{Agreement and disagreement opinion cascades. A. Schematic of the network structure. In a large network, only a small number of agents receive inputs, indicated by arrows. Inputs have the same strength, i.e., the same norm, but they can differ in the information they bring,  i.e., which option they favor. In the numerical simulations, inputs are in favor (blue arrows) or in disfavor (red arrows) of Option~1. B. Network response to inputs for cooperative (left) and competitive (right) agents. B1. Small inputs (no cascade). B2. Large inputs (cascade). C. Cascade threshold analysis. The plots show how the average network opinion strength changes as a function of the seed input amplitude (in logarithmic scale) over ten network instances with random seed placement. $\No=2$ in C1 and $\No=3$ in C2. Left plots: Watts-Strogatz network. Right plots: Barabasi-Albert network. Full lines: cooperative agents, $u_{th}=0.05$. Dashed line: competitive agents, $u_{th}=0.05$. Dotted lines: cooperative agents, $u_{th}=0.1$. Dashed-dotted lines: competitive agents, $u_{th}=0.1$. Other parameters as in Figures~\ref{fig:dis-agreement WS} (Watts-Strogatz case) and~\ref{fig:dis-agreement BA} (Barabasi-Albert case). In all simulations, $\delta_u=0.2$.}\label{fig:attention threshold}
\end{figure}

In the following, we fix the number $N_{seeds}$ and let  all inputs $\bs{b}_1,\ldots,\bs{b}_{N_{seeds}}$ have the same norm $\bar b$ but still allow %they either 
each to bring information either in favor or in disfavor of Option~1.  In the attention dynamics~\eqref{eq:udot}, we use $u_{min}=u_a-\delta_u$ and $u_{max}=u_a+\delta_u$, where $u_{a}$ is defined in~\eqref{eq:agrement threshold}, in the agreement (cooperative agents) regime, and $u_{min}=u_d-\delta_u$ and $u_{max}=u_d+\delta_u$, where $u_{d}$ is defined in~\eqref{eq:disagrement threshold}, in the disagreement (competitive agents) regime. This choice makes the opinion cascade behavior independent of the network-specific opinion-formation bifurcation point. The network topology, seed placement, seed input amplitude, and the attention threshold $y_m$ are thus the only determinants of the observed opinion cascade behavior.

To verify the existence of a network threshold for the trigger of opinion cascades, in Figure~\ref{fig:attention threshold} we systematically vary network topology, attention threshold, and seed input amplitude, for randomly placed seeds. Figures~\ref{fig:attention threshold}C1 and~~\ref{fig:attention threshold}C2 correspond to the case $\No=2$ and $\No=3$, respectively. All plots show the steady-state average network opinion strength $\frac{1}{\Na}\sum_{i=1}^\Na||\Zz_i^*||$ as a function of the seed input amplitude in the coupled~\eqref{EQ:Specialized Dynamics}-\eqref{eq:udot} system. Right plots are the average results for ten instances of a Watts-Strogatz network with randomly placed seeds and left plots are the average results for ten instances of a Barabasi-Albert network with randomly placed seeds ($N_{seeds}=5$ for both types of network). For each network and each value of $\No$, we consider four different regimes: cooperative agents, $y_m=0.05$ (full lines); competitive agents, $y_m=0.05$ (dashed lines); cooperative agents, $y_m=0.1$ (dotted lines); competitive agents, $y_m=0.1$ (dashed-dotted lines). In all cases one can detect a threshold value above which distinctively non-zero average opinion emerge, indicating a cascade as in Figures~\ref{fig:attention threshold}B1 and~B2.

\af{The cascade threshold increases monotonically with $u_{th}$. The parameter $u_{th}$ tunes cascade sensitivity at the single agent level, in the sense that smaller (larger) values of $u_{th}$ decrease (increase) the minimum input strength necessary to trigger a cascade at the single agent level~\citep[Theorem~VI.6]{bizyaeva2020general}. In the presence of networked interaction, the tunable cascade sensitivity of the single agent level scales up to the network level: for large $u_{th}$, large seed input strengths are necessary to trigger a cascade; for small $u_{th}$, small seed input strengths are sufficient to trigger a cascade.}

\af{Observe that in our model cascades can be partial or complete. For instance, the agreement cascade is complete in Figure~\ref{fig:attention threshold}B2 left because all agents develop a markedly non-zero opinion. Conversely, the disagreement cascade is only partial in Figure~\ref{fig:attention threshold}B2 right because some of the agents, including one of the seeds, possess a roughly neutral opinion at the post-cascade steady state. We further observe that, generically, when the bifurcation parameter is the seeds' input strength, as considered here, cascades happen through a simple saddle-node bifurcation. It follows that for supra-threshold inputs the time to reach the post-cascade steady state is inversely proportional to the input strength, a phenomenon known in the dynamical system literature as ``traversing the ghost of the saddle-node bifurcation"~\citep[Section~3.3.5]{Izhikevich2007}. A detailed analysis of both the transition from partial to complete cascades and the time to the post-cascade steady state is left for future developments}.

\section{Control of agreement and disagreement opinion cascade by input centrality assignment}
\label{sec: cascade control}

This section explores how the sensitivity of an opinion-forming network to trigger an agreement or disagreement cascade in response to a small number of inputs can be controlled by exploiting the emergent properties of the network. We introduce notions of agreement and disagreement centrality. These centrality indices can be used to guide the cascade seed placement towards maximizing or minimizing the network response to inputs, independent of the intrinsic node properties like, e.g., the attention threshold parameter $u_{th}$. In other words, our cascade model can be controlled independently and simultaneously in two ways: intrinsically at the agent level, by tuning the attention dynamic parameters, and extrinsically at the network level, by exploiting the network topology for efficient seed placement.  We consider extrinsic control %and examine the selection of seeds that enhance or diminish 
that enhances or diminishes sensitivity to inputs for agreement (disagreement) cascades by selecting as seeds the most or least central agents as defined by the agreement (disagreement) centrality index.  The selected set of agents is optimal in the case centrality is modular; otherwise, it provides only an approximation to the optimal set of agents (see, for example, \citep{ClarkSupermodular2013,Fitch2016}). We save the investigation of optimality for future work.

\subsection{Agreement and disagreement centrality}

A key prediction of Theorem~\ref{THM: dis-agree dicothomy} is that the absolute value of the entries of the adjacency matrix's eigenvectors $v_{max}$ and $v_{min}$ predict (up to linear order) the opinion strength (and sign) of each agent along agreement and disagreement bifurcation branches, respectively. This result suggests the following notion of centrality for agreement and disagreement opinion formation.

Let $\tilde A$ be the unweighted adjacency matrix of an undirected, connected network. Let $v_{max}$ be the normalized positive eigenvector associated to the unique maximal eigenvalue of $\tilde A$. Assume that the minimal eigenvalue is also unique and let $v_{min}$ be the associated normalized eigenvector.

\begin{definition}\label{def:centralities}
	We define {\em agreement centrality of agent $i$} as $[v_{max}]_i$,  {\em disagreement centrality of agent $i$} as $|[v_{min}]_i|$, and  {\em signed disagreement centrality of agent $i$} as $[v_{min}]_i$
\end{definition}

The proposed notion of agreement centrality is the same notion of eigenvector centrality introduced in~\cite{bonacich1972}. The proposed notion of disagreement centrality is new. Both notions of centrality will be key to analyzing and controlling opinion cascades over networks. It makes sense to talk about a {\it signed} disagreement centrality because the eigenvector $v_{min}$ has mixed sign entries (whereas $v_{max}$ can always be chosen to be positive). We will show that signed disagreement centrality is instrumental to controlling disagreement opinion cascades. For directed networks, one could use the eigenvectors of $\tilde A+\tilde A^T$ to define suitable centrality indices. We  leave these ideas for future investigations.

\subsection{Control of agreement opinion cascades by centrality assignment}
\label{ssec: agreement control}

Suppose that we can choose $N_{seeds}$ agents to receive input, and suppose that each input has the same norm $\bar b$, but brings potentially different information (i.e., favoring or disfavoring) about, say, Option~1\footnote{The proposed agreement cascade control algorithm naturally generalizes to the case in which inputs bring mixed information about the various options.}.

\begin{itemize}
	\item {\bf Towards maximization of agreement cascade sensitivity.} Let the control objective be to maximize the sensitivity of agreement opinion cascades to the $N_{seeds}$ inputs. Then an effective approach is to select the  cascade seeds to be the $N_{seeds}$ most central agents, according to the agreement centrality.
	\item {\bf Towards minimization of agreement cascade sensitivity.} Let the control objective be to minimize the sensitivity of agreement opinion cascades to the $N_{seeds}$ inputs. Then an effective approach is to select the cascade seeds to be the $N_{seeds}$ least central agents, according to the agreement centrality.
\end{itemize}

The rationale behind the proposed control strategy is the following. Maximally central agents are those that tend to develop stronger opinions at an agreement bifurcation. Triggering the network attention/opinion-strength positive feedback at maximally central agents therefore enhances the network response toward crossing the agreement cascade threshold. Likewise, minimally central agents tend to develop small opinions at an agreement bifurcation. Triggering the network attention/opinion-strength positive feedback at minimally central agents therefore diminishes the network response and keeps it away from crossing the agreement cascade threshold.

We illustrate the effectiveness of our agreement cascade control strategy in an eleven-agent network with $N_{seeds}=2$ and $\No=3$. Figure~\ref{fig: lowD control}A1 shows the network topology and the resulting node agreement centralities (visualized as the node sizes). The two black arrows point to the two most central agents. The two gray arrows point to the two less central agents. All agents have an intrinsic attention threshold $u_{th}=0.1$. When the two inputs affect the two most central agents following our sensitivity-maximization algorithm, inputs with amplitudes as small as $0.1$ are sufficient to trigger a full network cascade with large average opinion strength (Figure~\ref{fig: lowD control}B1 - black curve) and in which all agents become opinionated (Figure~\ref{fig: lowD control}C1, top). Conversely, when the two inputs affect the two less central agents following our sensitivity-minimization algorithm, the network is still resisting to trigger a full cascade even for input amplitude as large as $1.0$ (Figure~\ref{fig: lowD control}B1 - gray curve). Only the two agents receiving an input are opinionated, while the rest of the agents remain close to neutral (Figure~\ref{fig: lowD control}C1, bottom).

\begin{figure}
	\centering
	\includegraphics[width=0.8\textwidth]{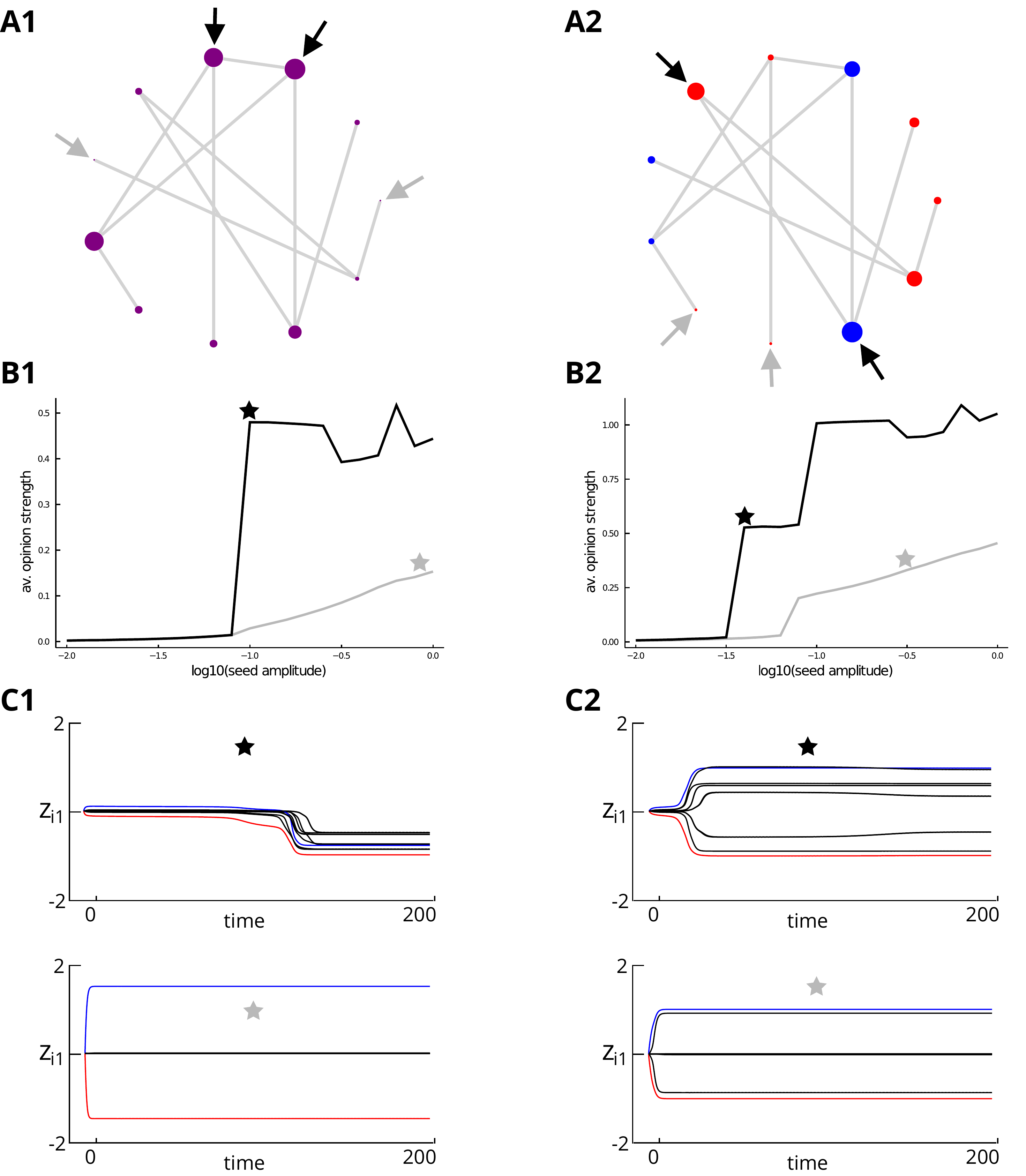}
	\caption{Control of agreement and disagreement opinion cascade by centrality assignment. A1: network topology and agreement centrality. The node size is proportional to its agreement centrality. A2: network topology and signed disagreement centrality. The node size is proportional to its (absolute) disagreement centrality. Blue nodes have positive signed disagreement centrality. Red nodes have negative signed disagreement centrality. B1: Average opinion strength as a function of the seed input amplitude under agreement cascade sensitivity maximization (black curve) and agreement cascade sensitivity minimization (gray curve). The network is in an agreement/cooperative regime. B2: Average opinion strength as a function of the seed input amplitude under disagreement cascade sensitivity maximization (black curve) and disagreement cascade sensitivity minimization (gray curve). The network is in a disagreement/competitive regime. C1: Evolution of agents' opinion about Option 1 in the agreement regime for the seed input amplitude indicated by the stars in B1. Top: under sensitivity maximization. Bottom: under sensitivity minimization. C2: Evolution of agents' opinion about Option 1 in the disagreement regime for the seed input amplitude indicated by the stars in B2. Top: under sensitivity maximization. Bottom: under sensitivity minimization. \af{Parameters: $\alpha=0.2$, $\beta=-0.5$, $\gamma=-\delta=0.1$ (agreement case), $\gamma=-\delta=-0.1$ (disagreement case), $u_{th}=0.1$, $n=5$, $\delta_u=0.2$}.  }
	\label{fig: lowD control}
\end{figure}

\subsection{Control of disagreement opinion cascades by centrality assignment}
\label{ssec: disagreement control}

The control strategy described in the previous section to control the sensitivity of agreement opinion cascade can be generalized to the disagreement case. In this case, we  exploit the notion of signed disagreement centrality introduced in Definition~\ref{def:centralities}. Suppose again that we can choose $N_{seeds}$ agents to receive input, and suppose that each input has the same norm $\bar b$, but brings potentially different information (i.e., favoring or disfavoring) about Option~1\footnote{The proposed disagreement cascade control algorithm naturally generalizes to the case in which inputs bring mixed information about the various options as long as ${\bs b}_{i}=\pm {\bs b}_0$, $i=1,\ldots,\Na$.}. The proposed control strategy is the following.

\begin{itemize}
	\item {\bf Towards maximization of disagreement cascade sensitivity.} Let the control objective be to maximize the sensitivity of disagreement opinion cascades to the $N_{seeds}$ inputs. Then an effective approach is to select  the cascade seeds as follows. Assign inputs bringing information in favor of Option~1 to the agents with the largest {\it signed} disagreement centrality. Assign inputs bringing information in disfavor of Option~1 to the agents with the smallest {\it signed} disagreement centrality.
	\item {\bf Towards minimization of disagreement cascade sensitivity.} Let the control objective be to minimize the sensitivity of disagreement opinion cascades to the $N_{seeds}$ inputs. Then an effective approach is to select cascade seeds to be  the $N_{seeds}$ least central agents, according to the (absolute) disagreement centrality index.
\end{itemize}

The rationale behind the proposed control strategy is the following. Two agents with large positive and large negative signed disagreement centrality, respectively,  naturally tend to develop strongly disagreeing opinions at a disagreement bifurcation. Injecting two discordant inputs at two such agents, on the one hand, will contribute to the triggering of the network attention/opinion-strength positive feedback and, on the other hand, will amplify the disagreeing nature of the two agents. Taken together, these two effects greatly enhance the network response towards crossing the disagreement cascade threshold. Likewise, minimally central agents tend to develop small opinions at a disagreement bifurcation. Triggering the network attention/opinion-strength positive feedback at minimally central agents therefore great diminishes the network response and keeps it away from crossing the disagreement cascade threshold.

We illustrate the effectiveness of our disagreement cascade control strategy on the same eleven-agent network used for the agreement case. Figure~\ref{fig: lowD control}A2 shows the resulting signed disagreement centralities, where red (blue) corresponds to negative (positive) signed disagreement centrality and the node size is proportional to the (absolute) disagreement centrality. As for the agreement case, our control strategy maximizes network sensitivity when inputs are placed on the most central agents (indicated by black arrows) according to our disagreement sensitivity-maximization algorithm, resulting in large disagreement cascades, with large average opinion strength (Figure~\ref{fig: lowD control}B2 - black curve) and all agents exhibiting strong opinions (Figure~\ref{fig: lowD control}C2, top), for inputs with amplitude as small as $10^{-1.4}\simeq 0.04$. Conversely, when inputs are placed on the less central agents (indicated by gray arrows) according to our disagreement sensitivity-minimization algorithm,  the network is still resisting to trigger a full disagreement cascade even for input amplitude as large as $10^{-0.5}\simeq 0.3$ (Figure~\ref{fig: lowD control}B2 - gray curve). All but four agents remain close neutral (Figure~\ref{fig: lowD control}C2, bottom).

\section{\af{Multi-Robot Task-Allocation Disagreement Cascades}}
\label{sec:task allocation cascades}

We now introduce attention dynamics~\eqref{eq:udot} into the distributed multi-robot task-allocation dynamics~\eqref{eq:multirobot_task_allocation} introduced in Section~\ref{sec:task allocation}. The benefit of introducing attention dynamics in the task-allocation algorithm is to let the robot engage in task allocation only in the presence of a sufficiently large urgency to accomplish the tasks. We assume that only a limited number of robots (e.g., those equipped with suitable sensors, or those located in suitable positions) can sense real-time changes in task urgency above or below the basal urgency levels $\mu_j$. We call urgency-sensing robots {\it zealous}. Zealous robots act as the seeds, or leaders, of task-allocation cascades.

An interesting property of the Frucht graph is that, despite all agents having the same degree, node disagreement centralities are highly heterogeneous, as illustrated in Figure~\ref{fig:task allocation cascade}A. In such a situation, our cascade-control strategy becomes crucial for efficient network placement of zealous robots.

\begin{figure}
	\centering
	\includegraphics[width=0.75\textwidth]{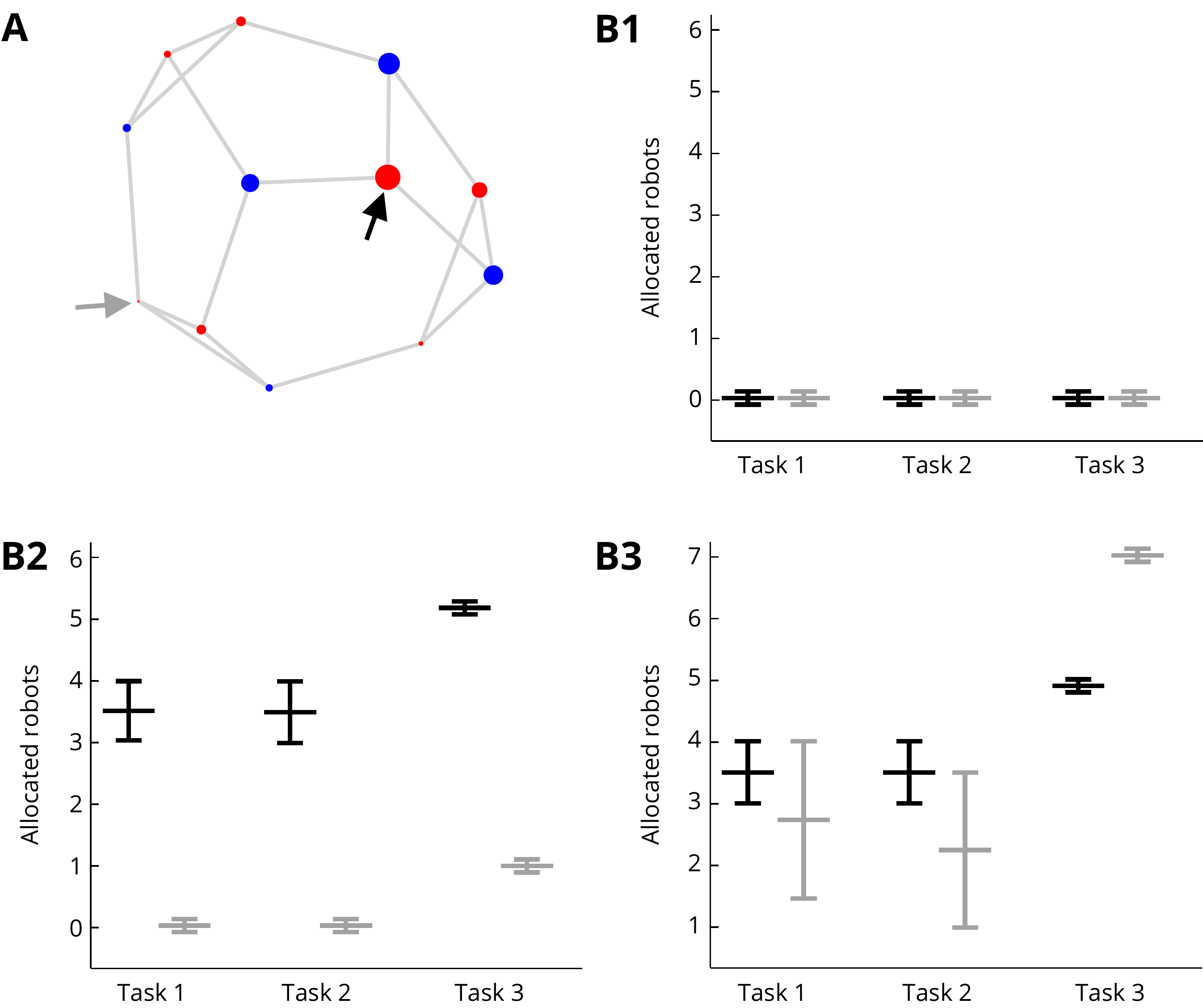}
	\caption{Distributed allocation over three tasks by robot swarm with attention dynamics, one zealous robot, and communication defined by the Frucht graph. The zealous robot perceives an increase in task urgency defined by $\rho_3$. A. The Frucht graph and its node signed disagreement centralities graphically represented as the node size and color (the larger the node, the larger its disagreement centrality; blue means positive signed disagreement centrality and red means negative signed disagreement centrality). The black (gray) arrow indicates the most (least) central node.  B.   Average (large bar) and standard deviation (short bar) of the number of agents allocated to Tasks~1,2,3 over one-hundred instances of our task-allocation cascade algorithm. Black: zealous robot is most central node. Gray: zealous robot is least central node. See main text for details. B1: $\rho_3=10^{-0.9}$. B2: $\rho_3=10^{-0.6}$. B3: $\rho_3=10^{-0.3}$. Basal task urgency: $\mu_1=\mu_2=0.3$ and $\mu_3=0.4$. Other parameters as Figure~\ref{fig:task allocation}.}.
	\label{fig:task allocation cascade}
\end{figure}

To illustrate this fact, we consider a robot swarm as described by Figure~\ref{fig:task allocation cascade}A and suppose that only one zealous robot is available, which detects an increase in the urgency of Task 3 of magnitude $\rho_3$. More precisely, if $i_z$ is the index of the zealous robot, then $\nu_{i_z 3}=3\rho_3$, %in such way that 
and, by~\eqref{eq:multirobot_task_allocation}, the effective urgency of Task~3 perceived by the zealous robot is $\mu_3+\rho_3$. Furthermore, all robots are equipped with the attention dynamics~\eqref{eq:udot}, with $u_{min}=u_d/2$, $u_{max}=2u_d$, $u_{th}=0.1$, $n=5$.

Figure~\ref{fig:task allocation cascade}B reproduces the model behavior in two cases: when the zealous robot is  the least central node (according to disagreement centrality) and when the zealous robot is the most central node (according to signed disagreement centrality). As in Figure~\ref{fig:task allocation}B, we run one-hundred simulations with random initial conditions and small random zealousness parameter perturbations for three different values of $\rho_3$, i.e., the urgency increase of Task 3 as perceived by the zealous robot. For small $\rho_3$ (Figure~\ref{fig:task allocation cascade}B1), %neither in the case of least centrally located or most centrally zealous robot the perceived increase in the urgency of Task~3 the zealous robot is sufficient to 
the zealous robot does not trigger a task-allocation cascade  when it is in the least central or most central location in the network. All robots remain neutral (i.e., not allocated). 

For intermediate values of $\rho_3$ (Figure~\ref{fig:task allocation cascade}B2), only when the zealous robot is in the most central location does the task-allocation cascade get triggered. Observe that the standard deviation of the number of agents allocated to Task~3 drops to zero and that the standard deviations of the number of agents allocated to Tasks~1 and~2 also drop drastically as compared to the model behavior without attention (Figure~\ref{fig:task allocation}B). Furthermore, the ratio between the average number of robots allocated to Task~3 and the average number of robots allocated to Task~1 or Task~2 is closer to the ratio of the tasks' basal urgency $\mu_3/\mu_1=\mu_3/\mu_2$ as compared to the model behavior without attention. We will analytically investigate these observations in future works. 

Finally, for large values of $\rho_3$ (Figure~\ref{fig:task allocation cascade}B3), a task-allocation cascade is triggered when the zealous robot is in the most central or least central location.  %the high-central zealous robot and the low-central zealous robot cases. 
However, only when the zealous robot is the most central node, do the robots allocate according to the task urgencies. When the zealous robot is the least central node, a disproportionate number of robots is allocated to Task~3. The allocation is not efficient. To summarize, disagreement centrality is key in placing urgency-sensing robots in a distributed task-allocation setting.

%% The following two can be remarks or short sentences
% \subsection{The cascade threshold ensures robustness to communication loss}

% {\bf new simulation with partial or complete communication loss}

% \subsection{Input strength and input-node centrality controls the time to cascade}

% Show the result with {\bf new simulation} (both for small and high sensitivity).

% Both short time and long time can be useful. Connection with speed/quality trade-off. Stress the difference between fixed sensitivity (i.e., the attention threshold parameter) and emergent sensitivity (i.e., input strength, input location).

\section{Application to large-scale complex networks}
\label{sec: complex networks}

In this section we provide numerical evidence that the analysis and control strategies introduced in this work scale up naturally to large networks with complex interactions. For comparison purposes, we focus on two types of complex networks: Watts-Strogatz with large rewiring probability and Barabasi-Albert. For the sake of illustration, we fix $\Na=100$ and $\No=2$ or $\No=3$.

\subsection{Centrality index predicts complex networks opinion formation behavior}

The agreement and disagreement opinion formation behavior in Watts-Strogatz networks (Figure~\ref{fig:dis-agreement WS}) exhibits quantitative differences as compared to Barabasi-Albert networks (Figure~\ref{fig:dis-agreement BA}). Despite possessing the same average degree and for the same set of parameters, in the Barabasi-Albert network (Figure~\ref{fig:dis-agreement BA}) agents tend to develop weaker opinions as compared to the Watts-Strogatz network (Figure~\ref{fig:dis-agreement WS}). Also, in the disagreement regime, only in the Watts-Strogatz network do the group opinions organize into markedly bimodal (polarized) distributions. These two observations can be explained and predicted in terms of the different centrality distributions in the two networks.

As predicted by Theorem~\ref{THM: dis-agree dicothomy}, at an opinion-forming bifurcation, the agent states are distributed along a one-dimensional kernel subspace according to the entries of the adjacency matrix's eigenvectors $v_{max}$ and $v_{min}$, corresponding to its maximum (agreement) and minimum (disagreement) eigenvalues, respectively. The opinion strength of agent $i$ along an agreement (resp. disagreement) branch is thus expected to be roughly (i.e., modulo nonlinear terms in the center manifold expansion) proportional to the absolute value of the $i$-th entry of $v_{max}$ (resp. $v_{min}$). Figure~\ref{fig:centr distr vs norm} illustrates this fact by plotting $|[v_{max}]_i|$ (resp.~$|[v_{min}]_i|$) against $||\Zz_i^*||$, i.e., the opinion strength of agent $i$ at the agreement (resp. disagreement) equilibrium reached for $u$ slightly above the agreement (resp. disagreement) bifurcation point~\eqref{eq:agrement threshold} (resp.~\eqref{eq:disagrement threshold}). In all cases, one can observe a monotone relationship between the two quantities, and the relationship is basically linear for the largest entries/opinion strengths. However, for such large networks ($\Na=100$ in the numerical examples), the high-dimensionality of the state space makes nonlinear effects noticeable in some cases, particularly at small entries/opinion strengths.
%Also, to test robustness of our theoretical predictions, all numerical examples are performed in the presence of small random inputs to all the agents. Even in a very high-dimensional setting like the present one and in the presence of perturbations, our theoretical prediction remains qualitativelyvalid.

\begin{figure}
	\centering
	\includegraphics[width=\textwidth]{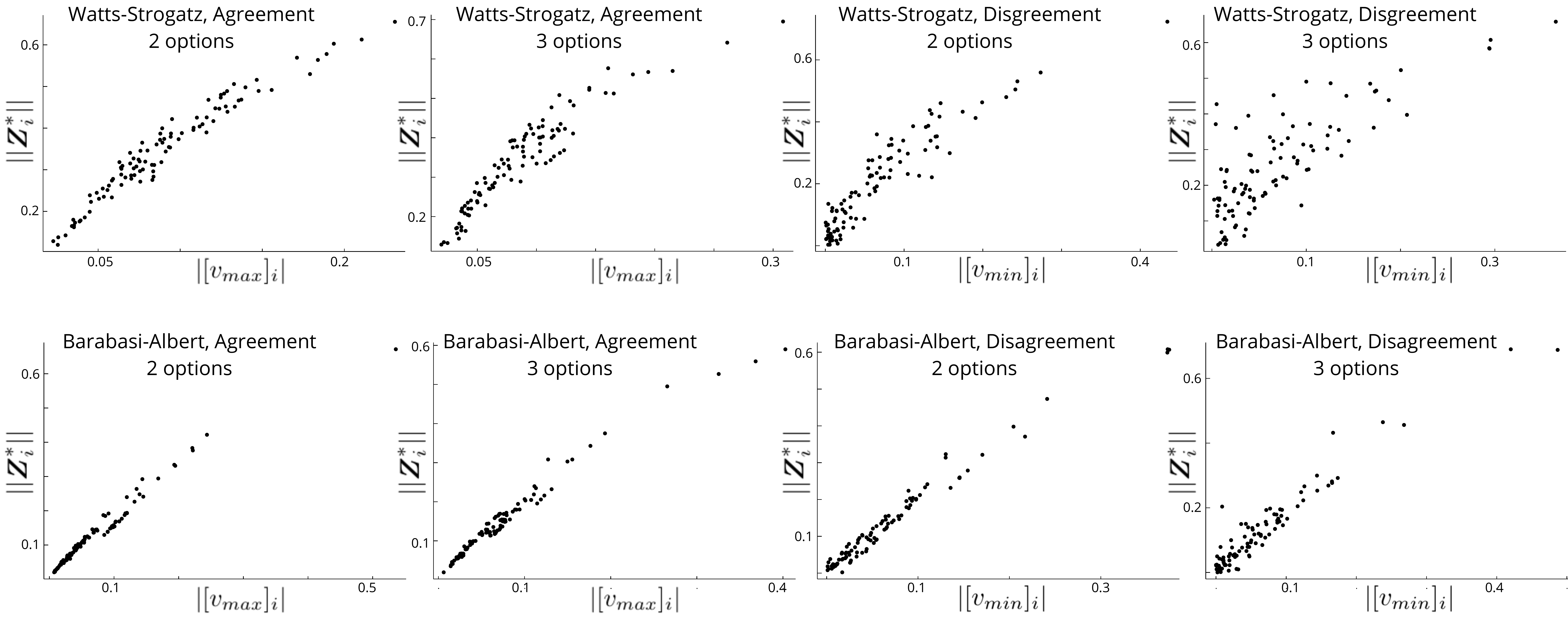}
	\caption{First and second column: Scatter plots of agent opinion strengths $||\Zz_i^*||$ (vertical axes) at agreement steady-states against the absolute value of the entries of the adjacency matrix's eigenvector $v_{max}$ (horizontal axes). Third and fourth columns: Scatter plots of agent opinion strengths $||\Zz_i^*||$ (vertical axes) at disagreement steady-states against the absolute value of the entries of the adjacency matrix's eigenvector $v_{min}$ (horizontal axes) Plot titles indicate the network type, if the bifurcation was of agreement or disagreement type, and the number of options.}
	\label{fig:centr distr vs norm}
\end{figure}

The quantitative difference between agreement and disagreement opinion formation on Watts-Strogatz and Barabasi-Albert networks can thus be understood (and predicted) by looking at the distribution of the absolute values of the entries of $v_{max}$ and $v_{min}$. Figure~\ref{fig:centr distr} shows the empirical distributions $f\left( \left|[v_{max}]_i\right| \right )$ and $f\left( \left|[v_{min}]_i\right| \right )$ of $|[v_{max}]_i|$ and $|[v_{min}]_i|$ obtained from one thousand realizations of Watts-Strogatz and Barabasi-Albert networks with the same size and parameters as those used in the simulation in Figures~\ref{fig:dis-agreement WS} and~\ref{fig:dis-agreement BA}. Both for $v_{max}$ and $v_{min}$, the distributions are heavier at smaller values and lighter at larger values for the Barabasi-Albert case than for the Watts-Strogatz case. Both the median and expected value are larger in Watts-Strogatz networks than in Barabasi-Albert networks. In other words, both agreement and disagreement opinion formation in Barabasi-Albert networks tend to be characterized by a larger (resp. smaller) number of agents developing weak (resp. strong) opinions than in Watts-Strogatz networks. For the disagreement case, this translates into the fact that disagreeing agents tend to polarize more easily over a Watts-Strogatz topology than over a Barabasi-Albert topology. 

\begin{figure}
	\centering
	\includegraphics[width=\textwidth]{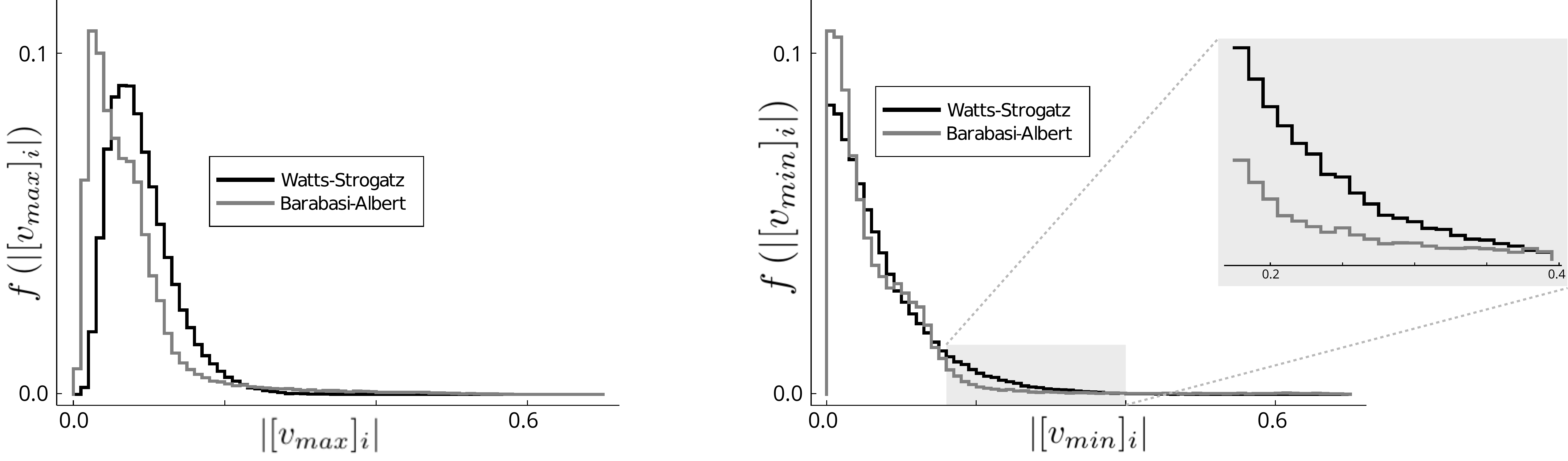}
	\caption{Empirical distributions of the absolute values of the entries of $v_{max}$ (left) and $v_{min}$ (right) for Watts-Strogatz (black curves) and Barabasi-Albert (gray curves) networks. %Plot titles indicate the type of network, if agreement or disagreement opinion formation was considered, and the number of options.
	}
	\label{fig:centr distr}
\end{figure}

\subsection{Centrality index predicts complex networks opinion cascade behavior}
\label{ssec: centrality and complex network analysis}

Figure~\ref{fig:attention threshold} highlights another quantitative difference between Watts-Strogatz and Barabasi-Albert networks.  
Only in Watts-Strogatz networks does the threshold behavior exhibit a marked all-or-none characteristic, in the sense that only for this kind of network does the average opinion switch from close to zero to a relatively constant non-zero value. This indicates that Watts-Strogatz networks are able to robustly trigger an opinion cascade as soon as the seed amplitude is above threshold, roughly independently of seed placement (which is random in the ten repetitions over which the results were averaged). Conversely, in Barabasi-Albert networks, random seed placement is incapable of robustly triggering a cascade. The averaged data does not exhibit a sharp threshold because different seed placements lead to sharply different cascade thresholds. As a consequence, the average opinion strength exhibits a strong linear dependence on the seed amplitude, reflecting the increasing fraction of cases in which even highly inefficient seed placements are able to trigger a cascade.

The quantitative differences between Watts-Strogatz and Barabasi-Albert networks can again be interpreted in  simple algebraic terms from their agreement and disagreement centrality distributions (Figure~\ref{fig:centr distr}). Because Barabasi-Albert networks tend to have smaller centrality values than Watts-Strogatz networks, random agents used as seeds tend to develop weaker opinions in the former than in the latter. As a consequence, it is harder to spread the attention threshold-crossing behavior in Barabasi-Albert networks than in Watts-Strogatz networks because, on average, the attention-opinion positive feedback loop is less active in the former than in the latter.

\af{Qualitatively, one could have predicted the same result by looking at more classical network properties, like the global clustering coefficient or the network diameter, which are both smaller in WS than BA networks. However, such statistical indices do not predict the patterns of opinion formation as our centrality indices do. In particular, they cannot distinguish between agreement and disagreement opinion formation and cascades. The crucial point of Theorem 1 is indeed that the bifurcation behavior at an agreement or a disagreement bifurcation is fully determined by the agreement and disagreement centralities, respectively. Local indices of clustering or shortest path length, like eccentricity, are also unsuitable to predict agreement and disagreement opinion formation and cascade behavior because (to the best of our knowledge) there are no general or simple connections between our notion of agreement and disagreement centrality and other local network properties.}
%We discuss the use of other centrality measures in Section~\ref{sec: discussion}.

\subsection{Control of opinion cascade in complex networks}

We finally illustrate the effectiveness of the control strategy introduced in Sections~\ref{ssec: agreement control} and~\ref{ssec: disagreement control} to maximize or minimize the sensitivity of agreement and disagreement opinion cascades to inputs on a large scale complex network. We use a Watts-Strogatz network in the case $\No=2$ and $\No=3$ and for two values of the intrinsic threshold $u_{th}$, which is assumed to be equal for all agents. Recall from Section~\ref{ssec:cascade tuneable} and Figure~\ref{fig:attention threshold} that the effect of increasing (resp., decreasing) $u_{th}$ is to decrease (resp., increase) the global sensitivity of the network cascade to inputs.

Figures~\ref{fig:dis-agreement control} top left (agreement, $\No=2$) and~\ref{fig:dis-agreement control} bottom left (agreement $\No=3$) show that an agreement opinion network with a global small sensitivity ($u_{th}=0.1$) can be made as sensitive as (or even more sensitive than) a network with relatively large global sensitivity ($u_{th}=0.05$) by using our control strategy to maximize the agreement cascade sensitivity of the former (dashed-dotted lines) and minimize the agreement cascade sensitivity of the latter (full lines). For all cases, there is roughly an order of magnitude between the seed amplitudes required to trigger a cascade in the same network controlled towards maximizing versus minimizing its sensitivity.

\begin{figure}
	\centering
	\includegraphics[width=\textwidth]{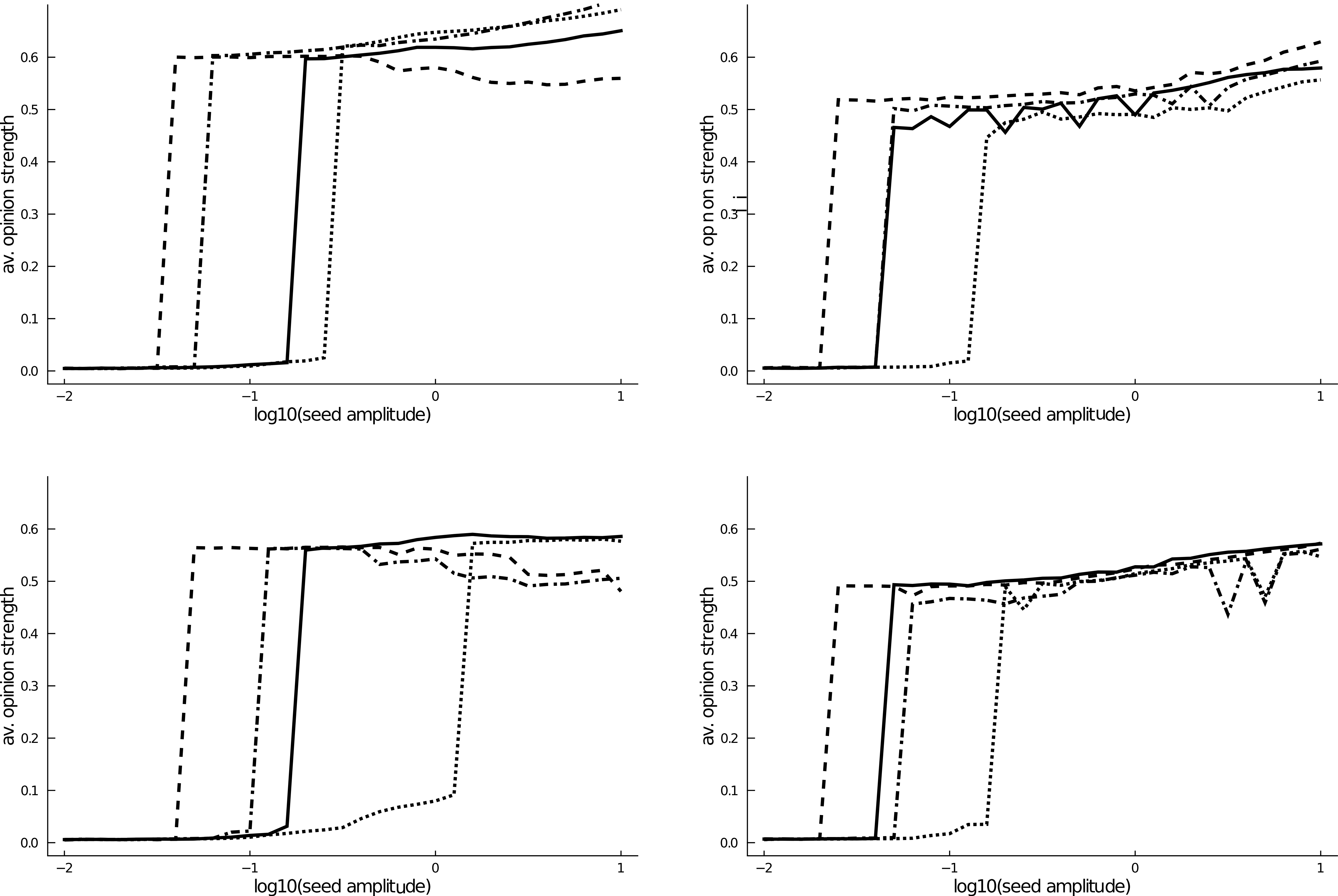}
	\caption{Control of agreement (left plots) and disagreement (right plots) opinion cascades for $\No=2$ (top) and $\No=3$ (bottom) in a Watts-Strogatz network with $\Na=100$ agents. Full and dashed lines correspond to $u_{th}=0.05$ (i.e., relatively large global sensitivity) with minimizing and maximizing sensitivity control, respectively. Dotted and dashed-dotted lines correspond to $u_{th}=0.1$ (i.e., relatively small global sensitivity) with minimizing and maximizing sensitivity control, respectively.}
	\label{fig:dis-agreement control}
\end{figure}

Figures~\ref{fig:dis-agreement control} top right (disagreement, $\No=2$) and~\ref{fig:dis-agreement control} bottom right (disagreement $\No=3$) show that also in the disagreement case a network with a global small sensitivity ($u_{th}=0.1$) can be made as sensitive as (or even more sensitive than) a network with relatively large global sensitivity ($u_{th}=0.05$) by using our sensitivity control strategy. As in the agreement case, for a fixed network topology, the seed amplitude required to trigger a cascade under sensitivity-maximizing control is much small (roughly half of an order of magnitude) than the seed amplitude required to trigger a cascade under sensitivity-minimizing control. \af{As discussed at the end of Section~\ref{ssec: centrality and complex network analysis}, we are not aware of local network properties other than our agreement and disagreement centrality indices that could provide such an efficient agreement and disagreement opinion-cascade control.}

\section{Discussion and future directions}
\label{sec: discussion}

\subsection{Summary of results}

The possibility of forming an opinion in favor or disfavor of different options makes opinion cascades intrinsically different from binary network cascades (e.g., spread of fads or information) usually analyzed and modeled in the literature. Motivated by these differences, we introduced two types of centrality for opinion cascades: agreement centrality  and disagreement centrality. The former predict the pattern of opinion formation when agents cooperate in the opinion formation process (i.e., when they tend to align their opinion to that of other agents) and provide an efficient, although heuristic and possibly sub-optimal, way to control the sensitivity of agreement opinion cascades by seed placement. The latter predict the pattern of opinion formation when agents compete in the opinion formation process (i.e., when they tend to reject other agents' opinions) and provide an efficient, although heuristic and possibly sub-optimal, way to control the sensitivity of disagreement opinion cascades by seed placement. We illustrated the practical validity of our theoretical predictions through an original example of distributed task allocation in multi-robot swarms and through numerical simulations in large random networks. Our centrality indices are grounded in the bifurcation behavior from a neutral to an opinionated state. To the best of our knowledge, this is the first time that centrality indices arise from intrinsically nonlinear bifurcation phenomena.

\subsection{\af{Agreement vs disagreement centrality}}

Our notion of agreement centrality is the same notion of eigenvector centrality introduced in~\citep{bonacich1972}. Our notions of disagreement centrality are new. Crucially, the two indices of centrality can be sharply different for the same graph. The Frucht graph used for the task-allocation example in Section~\ref{sec:task allocation} reveals a key difference between agreement and disagreement centrality: whereas the former cannot distinguish nodes in a regular graph (eigenvector centrality is the same for all nodes in a regular graph), the latter can. We thus suggest that disagreement centrality might turn out to be a key  to understanding complex networks where disagreement opinion formation is the norm, e.g., sociopolitical networks.

\subsection{\af{Simple and complex contagion}}

The seminal work~\citep{centola2007complex} revealed that the same network topology, in particular, in the presence of the bridges typical of small-world networks, can behave very differently during ``simple'' versus ``complex'' contagion. A simple contagion is one in which a single active neighbor is sufficient to spread the contagion to a node. A complex contagion is one in which multiple active neighbors are necessary to spread the contagion to a node. Because  time and state-space are continous in our model, the difference between simple and complex contagions becomes nuanced, i.e., contagions can simultaneously be simple and complex. In the presence of neighbors with sufficiently strong opinions the contagion is simple, because a single, strongly opinionated neighbor can be sufficient to let a given agent cross the cascade threshold. Conversely, if neighbors' opinions are weak the contagion is complex, because multiple opinionated neighbors are necessary to let a given agent cross the cascade threshold. The weaker the neighbors' opinions, the more complex the contagion. The results of this paper show that our notion of agreement and disagreement centrality are naturally suited to analysis and control of mixed, i.e., both simple and complex, contagions.

\subsection{Extensions and future directions}

\subsubsection{Joint centrality and connection with other centrality indices}

The proposed centralities exactly (up to nonlinear order in the center manifold expansion) characterize the pattern of agreement and disagreement opinion formation for {\it open-loop} changes of the attention parameter in the homogeneous system~\eqref{EQ:Specialized Dynamics} without inputs. As such, they have been shown to provide a powerful heuristic for cascade seed placement. However, whether the resulting seed-placement algorithm is optimal remains an open question.  A number of other centrality indices (e.g., information centrality~\citep{stephenson1989rethinking}, diffusion centrality~\citep{banerjee2013diffusion}) are available for comparison. Most importantly, the results presented in~\cite{ClarkSupermodular2013,Fitch2016} show that agent-selection based on centrality becomes subtle if more than one agent must be selected. In the case of two agents, simply selecting the first and second most central agent generically does not provide the optimal solution, at least in highly regular networks. In particular, the notion of joint centrality introduced in~\cite{Fitch2016} provides the tool to more rigorously tackle the multi-agent selection problem with respect to a given centrality index. We will explore how the notion of joint centrality applies to the proposed centralities in future works.

\subsubsection{Heterogeneous agent thresholds}

As in classical work on network cascades~\citep{granovetter1978threshold,schelling2006micromotives,watts2002simple,lim2015simple,garulli2015analysis,rossi2017threshold}, the parameters of the closed-loop attention dynamics, and thus the resulting single-agent cascade threshold, can be assumed to be heterogeneous and drawn from some random distribution. By the intrinsic robustness of hyperbolic dynamical systems organized by bifurcations~\citep{Golubitsky1985}, for small heterogeneity, we expect the the result derived in this paper to remain robustly true. How large random heterogeneity interacts with the underlying network structure is an open question.

\subsubsection{Non all-to-all intra-agent coupling}

Another important extension will be to consider other than all-to-all intra-agent coupling topologies in the homogeneous regime~\eqref{EQ:Specialized Dynamics}. All-to-all intra-agent coupling models the case in which options are {\it a priori} equivalent, with no hierarchies or clustering among them. Introducing more structured intra-agent coupling allows to group and order options into different issues, with or without hierarchies among them. This modeling option is particularly relevant when the opinion developed over an issue (e.g., at which restaurant to have dinner) influences and is influenced by opinion formation over another issue (e.g., which type of dish to eat). We will extend our modeling and analysis work to other than all-to-all intra-agent coupling in future work.

\subsubsection{Time-varying network topologies and communication loss}

We remark that, although not yet addressed by our theory, our modeling framework naturally allows to model and exploit time-varying network topologies. For instance, the network topology can change dynamically as a function of the agents' opinion state (as in classical bounded-confidence opinion-formation models~\cite{LiScaglione2013}), which can naturally lead to clustered networks with clusters of agents having similar opinions, e.g., clusters of robots accomplishing similar tasks (see~\cite[Theorem~III.5]{bizyaeva2020general} for a static topology clustering result).
%Our modeling approach thus opens the possibility of dynamically creating, reorganizing, and dispersing clusters of robots with similar tasks, working cooperatively as macro-agents interacting through a lower-dimension inter-cluster network (as in~\cite[Theorem~III.5]{bizyaeva2020general}).
\af{Time-varying network topologies can also arise from communication loss between the agents, e.g., due to faults or when robots move too far away from each other. Crucially, as detailed in Section~\ref{sec:attention dynamics}, our attention dynamics realizes a local multi-stable memory of the opinion developed during a cascade. This distributed memory might naturally be suited to make the opinion formation behavior robust to temporary loss of communication between two or more agent subpopulations.}

\subsubsection{\af{Reversal to a neutral state through excitable attention-opinion dynamics}}

In some occasions it might be useful for the swarm to revert from an opinionated state back to a neutral state, e.g., when a set of tasks or a foraging mission has been accomplished. This can be achieved by making the attention-opinion dynamics excitable, i.e., the transition to an excited (opinionated) state from the neutral state in response to inputs is only transient. Mimicking neural excitable dynamics~\citep{fitzhugh1961impulses,hodgkin1952quantitative}, the multi-stable attention-opinion dynamics present in each agent can be transformed into  excitable dynamics by adding a slow adaptation variable, providing slow negative feedback on the attention-opinion dynamics. Fast positive feedback acting in conjunction with slow negative feedback constitutes the basic motif of excitability~\citep{franci2018robust}. Depending on the timescale of the slow adaptation variable, a transition to opinion through a cascade is returned to neutral after a tunable delay, allowing  the swarm to become again reactive to new circumstances. The implementation of such an excitable attention-opinion dynamics will be the subject of future work.

\subsection{Applications beyond multi-robot task-allocation}

\subsubsection{A flexible best-of-$n$ {\it and} task-allocation model}

Our model can rapidly and reliably switch between agreement and disagreement behavior. As such, it constitutes a dynamical opinion formation model that, when applied to robot swarms, can rapidly and reliably switch between best-of-$n$ and task allocation. This flexibility can be key for adaptive behavior in robot swarms. For instance, during a foraging task \citep{zedadra2017multi,pitonakova2018information} the possibility to switch between exploration, when the robots use the opinion dynamics in the disagreement regime to allocate them to explore different directions, %starting from the nest, 
and exploitation, when the robots use the opinion dynamics in the agreement regime to reach a consensus about the richest source location, is key to adaptive capabilities for the swarm. The parameters $\gamma$ and $\delta$, which determine the agreement ($\gamma>\delta$) or disagreement ($\gamma<\delta$) regime, can be equipped with input and state-dependent dynamics as in~\citep[Section VI.D]{bizyaeva2020general}, which allows for flexible transitions between exploration and exploitation as a function of the information collected about the environment and the state of each robot.

Opinion cascades, both in a best-of-$n$ or a task-allocation setting, are useful when only a limited number of information-gathering agents are available. The information-gathering agents work as the seeds of an opinion cascade in the robot swarm. Through our threshold mechanism, a robot swarm can decide to remain in a neutral state (i.e., ``do nothing'' or ``stay charging at the dock'') until a sufficiently large amount of information has been gathered by the seeds. Past this threshold, opinion formation spreads through the network. The resulting opinion (i.e., task-allocation) state depends on the network topology, on the seed opinion states, and on possibly small (i.e., sub-threshold) biases of non-information gathering agents. Importantly, the seed-placement algorithm developed in Section~\ref{sec: cascade control} can be used to tune the swarm sensitivity to engage in the multi-task activity. When the sensitivity is small, the swarm behaves ``lazily,'', with relatively large inputs needed to trigger the transition to an active (i.e., opinionated) state. When the sensitivity is large, the swarm behaves ``hyperactively,'' with only relatively small inputs needed to trigger the beginning of the multi-task activity.

\subsubsection{Biological and artificial swarms in motion}

In both animal and robot swarms in motion, the decision to be made is about when and where to move. As discussed in~\cite{AF-MG-AB-NEL:20} our modeling framework naturally accommodates this situation, too, by letting the attention/bifurcation parameter be related to the geometry of the motion space, e.g., the group distance and relative position with respect to relevant objects like food sources, prey or predators, resources, or searched items. A number of works showed the relevance of similar approaches for  biological systems~\citep{Couzin2005a,Nabet2009,Couzin2011,Leonard2012,Pinkoviezky2018}. Collision avoidance~\citep{wang2017safety,van2011reciprocal} can also be modeled and tackled in our opinion formation framework. A neutral opinion means in this setting moving straight, favoring Option~1 means turn right and favoring Option~2 means turn left. If in agreement regime, two colliding robots will either both turn right or turn left and collision will be avoided. In this simple setting, the attention parameter dynamics can be a function of the robots' relative positions.

Similarly to the multi-task setting, the introduction of opinion cascades allows to model and design agent swarms in motion in which only a small number of agents (i.e., those that can sense objects and items in the environment) gather information about the geometry of the motion space. The same emergent behaviors described in the multi-task setting (e.g., tunable sensitivity of the swarm to informative inputs) are to be expected in the motion setting.

\subsubsection{Sociopolitical networks and political polarization}

The notion of ``opinion'' most naturally applies to networks of people discussing, expressing, and forming ideas (and opinions) about sociopolitical issues (e.g., governance, economics). Models of opinion dynamics have a long history~\citep{DeGroot1974,FriedkinJohnsen1999,CisnerosVelarde2019PolarizationAF,OlfatiSaber2004,Altafini2013}. Many modeling efforts focus on understanding the origin of polarized political views~\citep{McCarty2019,Galam1991,Macy2003,Axelrod1997,Dandekar2013}. Our modeling approach  generalizes many existing opinion-formation models~\citep{bizyaeva2020general} and provides the means towards a unified theoretical viewpoint on political polarization~\cite[Section~3.3]{AF-MG-AB-NEL:20}.

Results from a recent experiment~\citep{macy2019opinion} support a ground-breaking hypothesis about the origin of polarized political views. Namely, in a network of initially neutral agents with heterogeneous political predispositions, % (Democratic or Republican) 
political polarization arises from opinion cascades triggered by early movers (i.e., the cascade seeds)   with the final opinion state unpredictable from the original political predispositions. Our model provides a rigorous explanation of these observations. The pattern of opinion formation reached at the end of a disagreement cascade is largely determined by the complex network topology and the seed placement, more than the intrinsic biases (i.e., political predispositions) of the agents. For fixed agent biases, different random distribution of seeds or changes in the random network topology can lead to completely different opinion formation patterns at the end of the cascade in our model. We are currently working on interdisciplinary sociopolitical research exploring these and other ideas.

%\begin{acknowledgements}
%\end{acknowledgements}

\section*{Conflict of interest}
The authors declare that they have no conflict of interest.

\bibliographystyle{myplainnat}
\bibliography{references}   % name your BibTeX data base

\end{document}